\documentclass[a4paper, twoside]{article}
\usepackage{mathrsfs}
\usepackage{graphicx} 
\usepackage{authblk}
\usepackage{amsfonts}
\usepackage{amsmath}
\usepackage{amssymb}
\usepackage{mathtools}
\usepackage{amsthm}
\usepackage{mathrsfs}
\usepackage{tikz}
\usepackage{fancyhdr}
\usepackage{fix-cm}

\usepackage[dvipsnames,table]{xcolor}
\usepackage{multirow}

\usepackage{titlesec}
\titleformat{\subsubsection}[runin]{\normalfont\bfseries}{}{0em}{}[]
\titlespacing*{\subsubsection}{0pt}{3.25ex plus 1ex minus .2ex}{0.5em}

\usepackage{geometry}
\usepackage[
hyperfootnotes=false,
colorlinks=true,
breaklinks=true,
pdfstartview=FitH,
pdfpagelayout=OneColumn]{hyperref}

\usepackage{color}

\newtheorem{theorem}{Theorem}[section]
\newtheorem{corollary}{Corollary}[section]
\newtheorem{lemma}{Lemma}[section]
\newtheorem{definition}{Definition}[section]

\usepackage[
backend=biber,
style=alphabetic,
sorting=nyt,
maxnames=4,
maxalphanames=4,
giveninits=true,
isbn=false,
url=false
]{biblatex}

\newcommand{\ket}[1]{|#1\rangle}
\newcommand{\bra}[1]{\langle #1|}
\newcommand{\op}[2]{|#1\rangle\langle #2|}
\newcommand{\ip}[2]{\langle #1| #2\rangle}
\newcommand{\samp}[0]{\leftarrow_{\$}}
\newcommand{\mc}{\mathcal}

\newcommand{\mbb}{\mathbb}
\newcommand{\mbf}{\mathbf}
\newcommand{\mrm}{\mathrm}
\newcommand{\msf}{\mathsf}
\usepackage{bm}
\newcommand{\X}{\mathsf{X}}
\newcommand{\Y}{\mathsf{Y}}
\newcommand{\Z}{\mathsf{Z}}
\newcommand{\tgt}{\textup{tgt}}
\newcommand{\src}{\textup{src}}

\newcommand{\D}{\textup{D}}
\newcommand{\tr}{\textup{Tr}}
\newcommand{\f}{\textup{F}}
\newcommand{\tv}{\textup{TV}}
\newcommand{\bc}{\textup{BC}}
\newcommand{\td}{\textup{TD}}
\newcommand{\id}{\textup{id}}

\newcommand{\guess}{\textup{guess}}

\title{Quantum Secure Non-Interactive Reductions}
\author[1,$\dagger$]{Maxwell Gold}
\author[1]{Sarah Hagen}
\affil[1]{Department of Physics, University of Illinois Urbana-Champaign}
\author[2]{Daniel Alabi}
\author[2]{Eric Chitambar}
\affil[2]{Department of Electrical and Computer Engineering, University of Illinois Urbana-Champaign}
\affil[$\dagger$]{mjgold2@illinois.edu}
\date{\today}

\begin{document}

\maketitle

\begin{abstract}
    Efficient models for secure computation often rely on \textit{offline} preprocessed correlations, which subsequently enable private \textit{online} computation from minimal assumptions. Within these models, the study of secure non-interactive reductions (SNIR) formalizes when one classical correlation can be non-interactively transformed into another, while guaranteeing information-theoretic simulation-based privacy. In this work, we introduce quantum secure non-interactive reductions (QSNIR), a natural extension of SNIR, in which the source and target resource are each a bipartite quantum state. Specifically, when the target resource is classical, our framework enables us to analyze the problem of disseminating private correlations from entanglement, in addition to classical resources. Our primary technical result is that QSNIR privacy in this context can be computed exactly via a semidefinite program (SDP) pair that quantifies an error operationally rooted in a simulation-based cheating advantage. We prove that QSNIR privacy is lower-bounded by a decision problem related to minimum-error state discrimination (MED), and use both definitions to compute exact one-shot privacy errors for universal two-party computation (2PC) preprocessing correlations.
\end{abstract}

\section{\label{sec:INT}Introduction}
Private computation from \textit{preprocessed} correlated randomness serves as a powerful methodology for constructing information-theoretic (or provable) malicious-secure multiparty computation (MPC)~\cite{Beaver-1992-EM,Bendlin-2011-SE, Damgard-2012-MC}. 
In a purely classical setting, however, it is impossible to disseminate these resources privately, without resorting to strong assumptions, such as a trusted third party and private channels, or computational hardness~\cite{Impagliazzo-1989-Of}.  

In contrast, quantum key distribution (QKD) allows a pair, Alice ($\mbf{A}$) and Bob ($\mbf{B}$), to disseminate keys, in a way that offers malicious-secure privacy \cite{Bennett-1984-Qc, Ekert-1991-Qc, Vazirani-2014-Fd}, and is composable \cite{Ben-Or-2005-UC, Muller-Quade-2009-cq}. 
The essential components for proving security come from the eavesdropper's inability to clone an unknown state~\cite{Bennett-1984-Qc}, or to spoof certain uniquely quantum \textit{nonlocal} correlations~\cite{Vazirani-2014-Fd}. 
Symmetric key can also be faithfully embedded in a quantum resource, such as an entangled Bell pair~\cite{Ekert-1991-Qc}.

Key is simple to derive from entanglement precisely because it is symmetric: $\mbf{A}$ wants to know exactly $\mbf{B}$'s sample, and vice versa. If instead we consider correlated randomness that is not constrained in this way, embedding it in a quantum system ``privately'' entails additionally proving that $\mbf{A}$ or $\mbf{B}$ cannot cheat via their quantum side-information, and learn the other's sample beyond what they can infer from an ideal sampling of the correlation. This makes many foundational quantum cryptographic primitives impossible~\cite{Mayers-1997-US,Lo-1997-Iq,Lo-1997-QB}.
In fact, all non-trivial 2PC primitives embedded in a quantum resource will leak against a malicious adversary~\cite{Salvail-2009-PT}.

In~\cite{Salvail-2009-PT}, a quantum resource that correctly instantiates a primitive is represented by a \textit{coherent embedding} of the preprocessing correlation. This embedding is a pure quantum state shared between $\mbf{A}$ and $\mbf{B}$, whose complex amplitudes have a modulus determined by the encoded distribution.\footnote{Importantly, their definition of leakage is minimized when the set of quantum states in which the distribution is embedded forms an orthonormal basis. Including any additional quantum degrees of freedom that are jointly tied to the embedded distribution can only increase the leakage (see Lem 4.3 in~\cite{Salvail-2009-PT}).} 
With this embedding shared between the parties, instantiating the primitive (i.e. sampling from the preprocessing correlation) is modeled by a non-interactive \textit{reduction}, where $\mbf{A}$ and $\mbf{B}$ both measure their part of the embedding. The precomputed embedding is assumed to be ideal, and the cheating party's quantum side-information is their reduced state, given the measurement of the honest party. The definition for privacy (or conversely, leakage), given in~\cite{Salvail-2009-PT}, is then based on the difference between the quantum mutual information in the classical-quantum (cq) state that arises from measuring half of the embedding, and the classical mutual information in the ideal correlation. While this definition provides necessary and sufficient conditions for perfect privacy, it lacks an operational interpretation when the leakage is non-zero, and the leakage per copy applies only in the asymptotic limit, such that it cannot be easily shown to compose.

In this work, we detail a simulation-based framework for quantum secure non-interactive reductions (QSNIR). This builds off of the classical paradigm of secure non-interactive reductions (SNIR) \cite{Agarwal-2022-SN}, and adapts it to a setting where correlations can arise from quantum entanglement in addition to classically correlated quantum resources. We develop a definition for QSNIR that handles this distinction and supports composition. 

When applied to the coherent embeddings, our definition yields the same constraints for perfect privacy as in~\cite{Salvail-2009-PT}, but is operationalized by a guessing game: A distinct global observer, referred to as the \textit{environment}, must distinguish between a real world, which contains the embedding, and an ideal world. In the ideal world, a \textit{simulator}, a local channel on the side of the cheating party, works against the environment, by attempting to invert the measurement that the cheating party would make if they were honest. Our central result is that QSNIR privacy, when applied to these classical-quantum reductions, is represented by a strongly dual SDP pair, which captures this operational guessing game, and can be used to compute exact one-shot leakages, referred to as \textit{privacy errors}, for any cryptographic primitive. In addition, we prove that the optimum of this SDP is lower-bounded by another SDP that computes privacy related to a guessing game based on minimum-error state discrimination (MED), which we term MED privacy.

Beyond these results, we explicitly compute one-shot optimal QSNIR and MED privacy errors for common universal 2PC correlations, such as \textit{oblivious key}~\cite{Wolf-2004-OT}. Proving each privacy error involves exploiting various symmetries of how the correlations are encoded in a complex Hilbert space, and may merit its own technical interest. Additionally, even over the small range of correlations we consider, we see a diverse slate of optima that warrants a deeper understanding of the gap between QSNIR and MED privacy, which is shown not to be tight. The existence of this gap is related to Blackwell's work on informativeness of experiments~\cite{Blackwell1951}. Centrally, in QSNIR privacy the actions of the simulator are a \textit{garbling} of the ideal correlation, in order to conflate the ideal experiment with the real one. MED privacy, within this context, is then a decision problem that looks at a subset consisting of only the real world experiments.\footnote{We return to this topic in Sec.~\ref{sec:DSC}.} 

\subsection{\label{sec:INT-RW}Related Work}

\subsubsection{Preprocessing correlations for MPC.}
 Kilian~\cite{Kilian-1988-Fc} first proved that MPC in its entirety can be founded on primitives for oblivious transfer (OT). In turn Wolf and Wullschleger~\cite{Wolf-2004-OT} showed how to preprocess OT from a correlation known as \textit{oblivious key}. Oblivious key is synonymous with another universal preprocessing correlation, Beaver's \textit{multiplication triples}, which directly enable private multiplication in secret sharing-based MPC~\cite{Beaver-1991-Sm}. Rabin~\cite{Rabin-2005-HE}, as well as Cr\'{e}peau and Kilian~\cite{Crépeau-1988-ET,Crepeau-1988-Ao}, alternatively demonstrated weakened versions of these assumptions, where MPC is reduced to correlations arising from \textit{symmetric} (bit-flip) and \textit{erasure} channels.

\subsubsection{Secure non-interactive reductions.}
\textit{Secure non-interactive reductions} (SNIR)
and \textit{secure non-intera\-ctive simulations} (SNIS) are security frameworks introduced to capture when one correlation can be locally transformed into another, while preserving simulation-based privacy~\cite{Khorasgani-2020-SN,Agarwal-2022-SN}. These works build upon (non-secure) \textit{non-interactive simulation} (NIS)~\cite{Gacs-1973-Ci,Witsenhausen-1975-SP,Wyner-1975-ci,AGKN14,KA16,DMN18,GKS16,STW20}, and establish definitions for correctness and privacy analogously to secure computation, in the non-interactive setting. Similarly, in contrast to \textit{pseudorandom correlation generators} (PCGs)~\cite{BCGI18,BCGIKS19,BCGIKS20}, these frameworks establish security guarantees that are information-theoretic, making no computational assumptions. Agarwal et al.~\cite{Agarwal-2022-SN} studied SNIRs via spectral characterizations, identifying which correlations admit secure reductions and proving separations from non-secure counterparts. Independently (and concurrently), Amini Khorasgani et al.~\cite{Khorasgani-2020-SN} analyzed feasibility and rate regions for secure non-interactive simulation, demonstrating fundamental limits on the efficiency of secure reductions. Our contributions directly extend this line of work by introducing QSNIR, where the source resource can be quantum
entanglement.

\subsubsection{Limits for two-party quantum cryptography.}
Provably secure quantum cryptography is impossible for most 2PC primitives. This includes important foundational tasks such as bit commitment~\cite{Lo-1997-QB, Mayers-1997-US}, OT~\cite{Lo-1997-Iq}, and fair coin tossing~\cite{Mochon-2007-Qw}. Salvail et al.~\cite{Salvail-2009-PT} showed that all non-trivial 2PC primitives necessarily leak information against an adversary who \textit{purifies} their actions with unbounded resources. While their work rigorously confines the correlations that can be extracted from quantum resources with perfect privacy to those that satisfy certain \textit{zero-error} entropic conditions~\cite{Wolf-2004-Zi}, the  definition for leakage they give is based on asymptotic entropic quantities and lacks a direct operational interpretation. Our work improves on this result by giving quantitative definitions of leakage that 1.) are zero when the same zero-error entropic properties hold, and 2.) are operationally ascribed to one-shot decision-theoretic problems about channel recoverability and state discrimination.

\subsubsection{Quantum MPC from computational assumptions.}
As of late, founding quantum advantage for MPC has necessitated rejecting entirely information-theoretic results~\cite{Dulek-2020-SM,AQY22,BKS23}. Dulek et al.~\cite{Dulek-2020-SM} extended the preprocessing model of \cite{Bendlin-2011-SE,Damgard-2012-MC} in building multiparty quantum computing (MPQC). Here, the authors assume a completely classical offline phase, instantiated with computational assumptions, in order to demonstrate how MPQC is possible with a dishonest majority and an information-theoretic online phase. 
Ananth et al.~\cite{AQY22} studied cryptography from the assumption of pseudorandom quantum states. Their result is again explicitly computational: from pseudorandom states they obtain commitments and, consequently, maliciously secure MPC in the dishonest majority setting. 
Finally, Bartusek et al.~\cite{BKS23} demonstrated that shared EPR pairs, together with computational assumptions, suffice for strong primitives such as one-shot string OT and two-round MPC. This work explicitly assumes sub-exponential LWE for one-shot OT and uses the quantum-accessible random oracle model for their two-round MPC result. 
Our work is complementary to all of these results. We identify an information-theoretic reduction framework to show when an ideal entangled resource is equivalent to an ideal preprocessing correlation, without relying on pseudorandomness, one-way functions, or related computational assumptions.

\subsection{\label{sec:INT-PR}Primary Contributions}

\subsubsection{Defining non-interactive simulation-based reductions between quantum resources.}
We adapt the security framework of SNIR (and equivalently SNIS) to the quantum setting, which we term QSNIR, and provide rigorous distance-based definitions for correctness and simulation-based privacy. We prove several important properties of QSNIR, including parallel and sequential composability, and an equivalence between correctness and privacy under local unitary transformations. Any SNIR can be translated to a QSNIR between mixed states, but our framework further supports secure reductions from classical correlations to entanglement. 

\subsubsection{Operationalizing leakage in perfectly correct reductions from classical correlations.}
When the target state is classical, a QSNIR represents the problem of securely sampling from a target distribution, now with entanglement as a valid initial resource. The goal is to show that an adversary who possesses only a part of the initial resource learns no more about the target correlation, beyond what they could if the sample was instead generated by a trusted third party. With entanglement as a resource, we show that when a reduction produces samples of the correct correlation when both parties are honest, we can compute any leakage exactly via an SDP, which directly encodes the definition for single-copy privacy. 

\subsubsection{Computing optimal one-shot privacy errors for 2PC correlations.}
We apply our definition of QSNIR to reductions from universal 2PC correlations to the entangled states which coherently embed the relevant distribution. This lets us compute exact single-copy privacy errors for these correlations, which compose. To prove the optimality of each error, we exploit different symmetries that arise in the source entangled states, as a result of the structure of the embedded distribution.  

\subsection{\label{sec:INT-OC}Organization of Contents}
In Sec.~\ref{sec:OVR}, we give a technical overview of our results, before presenting notation and preliminaries in Sec.~\ref{sec:PRE}.  Next, in Sec.~\ref{sec:QSNIR}, we define the framework for QSNIR, and prove its composability. We then use this definition in Sec.~\ref{sec:BND}, and derive results surrounding privacy in reductions from classical correlations to their canonical coherent embeddings. We thereafter apply these results in Sec.~\ref{sec:2PC} to compute explicit one-shot optimal privacy errors for 2PC correlations.

\section{\label{sec:OVR}Technical Overview}
\subsection{Definitions of Properties of QSNIR}
We develop a simulation-based framework for secure non-interactive reductions from one bipartite quantum resource to another. Within this framework, we term each reduction a quantum secure non-interactive reduction (QSNIR). For a pair of parties $\{\mbf{A},\mbf{B}\}$, the goal of a QSNIR is to demonstrate that a sufficient requirement for distributing a target state, $\rho^{\mbf{AB}}_{\tgt}$, is to first distribute a source state, $\rho^{\mbf{AB}}_{\src}$, without any additional resources. 
\begin{definition}[$(\epsilon,\delta_{\mbf{A}},\delta_{\mbf{B}})$-QSNIR.  Informal, see Def.~\ref{def:eQSNIR}]
    An $(\epsilon,\delta_{\mbf{A}},\delta_{\mbf{B}})$-\textup{QSNIR} is a pair of local channels which transforms $\rho^{\mbf{AB}}_{\src}$ to $\rho^{\mbf{AB}}_{\tgt}$, with $\epsilon$-correctness, $\delta_{\mbf{A}}$-privacy when $\mbf{B}$ is honest, and $\delta_{\mbf{B}}$-privacy when $\mbf{A}$ is honest. 
\end{definition}
\noindent
In Sec.~\ref{sec:QSNIR}, we give corresponding distance-based definitions for these security parameters, where the goal is to minimize the distance between realizations of the problem in distinct real and ideal worlds. 
Any SNIR can be translated to a QSNIR between mixed states. The converse is generally not true as entanglement has no representation as a classical correlation, which can be seen from the existence of certain nonlocal correlations~\cite{Clauser-1969-PE,Fritz-2012-Bt}. 

With our general definition for QSNIR, we show that parallel and sequential composition is supported.
\begin{theorem}[Composition. Informal, See Thms.~\ref{thm:parallel} and~\ref{thm:sequential}]
    \textup{QSNIR}s support parallel and sequential composition, with at worst additive error.
\end{theorem}
\noindent
In the context of a preprocessing phase for 2PC, parallel and sequential composition allows us to bound the total error when extracting multiple correlations from product quantum resource states and employing these extracted correlations in further computation. 

\subsection{Computing QSNIR Privacy from Perfect Correctness}
When the target state is a classical correlation, $\rho_{\tgt}^{\X\Y}\sim(\X,\Y)$, and we fix the source state to be pure, a QSNIR details the problem of securely extracting classically correlated randomness from entanglement. When $\epsilon=0$ in this setting, we have a QSNIR with perfect correctness, which is well modeled by a reduction from $(\X,\Y)$ to a state, $\ket{\mrm{\Psi}_{\X\Y}}^{\mbf{AB}}$, representing one of its coherent embeddings (see Def.~\ref{def:coh-embed}). When both parties measure their part of the embedding in the computational basis, they obtain of a sample of the correlation. These measurements are described on average by the completely dephasing channel. Fig.~\ref{fig:CQSNIR}(a) formally depicts this reduction.

Here we prove the necessary and sufficient conditions for this reduction to also have perfect privacy.
\begin{theorem}[Perfect Classical-QSNIRs. Informal, see Thm.~\ref{thm:perf-priv}]
    A reduction from a correlation to its coherent embedding has perfect privacy, if and only if the correlation is trivial.
\end{theorem}
\noindent
Triviality refers to constraints on entropies for the common and dependent random variables constructed from the correlation, which we discuss in Sec.~\ref{sec:PRE-com-dep}.

When the correlation is nontrivial, $\delta>0$ necessarily. Interestingly, we can lower bound this $\delta$ by constructing an alternative definition for privacy based on MED optimal guessing probabilities. Here, we define this error, denoted $\delta_{\textup{MED}}$, as the difference between the optimal guessing probability for the adversary given either a sample of the ideal correlation, or access to the reduced state ensemble of the embedding, $\{\ket{\psi_{x}}\}_{x}$, which they see in the real world. A diagram for this privacy game is presented in Fig.~\ref{fig:CQSNIR}(b).
\begin{lemma}[MED Privacy $\leq$ QSNIR Privacy. Informal, see Lem.~\ref{lem:min-err-LB}]
    For any reduction from $(\X,\Y)$ to its embedding $\ket{\mrm{\Psi}_{\X\Y}}^{\mbf{AB}}$, $\delta_{\textup{MED}}\leq\delta_{\textup{QSNIR}}(\eqqcolon \delta)$ always holds.
\end{lemma}
\noindent
We prove this lemma via a continuity inequality for the conditional minimum-entropy of the cq states in the real and ideal worlds. To compute $\delta_{\textup{MED}}$ we can utilize the same strongly dual SDP that is used to compute the optimal guessing probability for MED (see Eq.~\ref{eq:sdp-med}).

We also demonstrate how the optimal $\delta_{\textup{QSNIR}}>0$ can be computed exactly via its own strongly dual SDP pair (see Eq.~\ref{eq:sdp-qsnir}). This SDP provides an operational interpretation to privacy related to our simulation-based framework. The task in this framework can be posed as follows, given a corrupt $\mbf{B}$: In the real world $\mbf{A}$ obtains a sample $x$ and $\mbf{B}$ obtains a reduced state $\ket{\psi_{x}}$ encoded on their registers. They output both resources, which together form a cq state, to a third party, the environment. In the ideal world $\mbf{A}$ again obtains $x$, but $\mbf{B}$ now obtains a state prepared by their simulator conditioned on $y$, and again they output the cq state to the environment. Privacy then quantifies the ability of this environment to distinguish the two cq states on average, which is exactly computed by the SDP pair that we present. This game is depicted in Fig.~\ref{fig:CQSNIR}(c). 

\begin{figure}[t]
    \centering
    \includegraphics[width=0.9\linewidth]{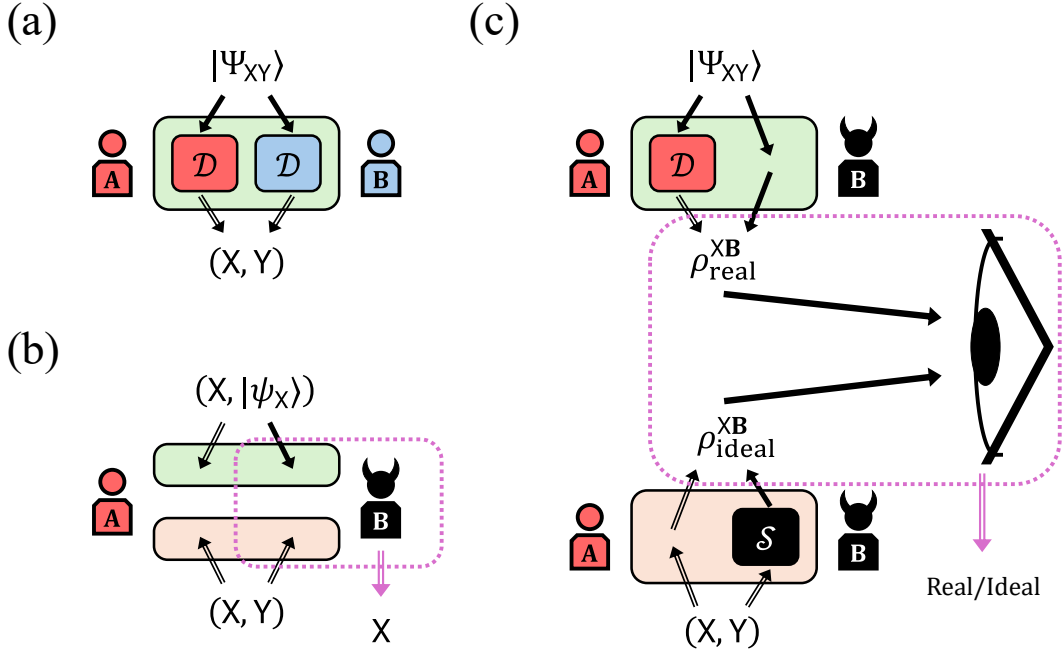}
    \caption{
    A QSNIR from a classical correlation, $(\X,\Y)$, to its coherent embedding, $\ket{\mrm{\Psi}_{\X\Y}}$, shared by parties $\{\mbf{A},\mbf{B}\}$. 
    (a) Perfect correctness via dephasing. 
    (b) The MED privacy game. 
    (c) The QSNIR privacy game. 
    Single-lined arrows denote the flow of quantum information, while double-lined arrows denote classical information. Real world operation are contained within green boxes, while ideal world operations are those in orange boxes.
    }
    \label{fig:CQSNIR}
\end{figure}

\subsection{Optimal QSNIR Privacy Errors for 2PC Correlations}
We apply our results for computing $\delta_{\textup{MED}}$ and $\delta_{\textup{QSNIR}}$ to several common preprocessing correlations for 2PC. This includes binary oblivious key (denoted $\msf{OK}$), as well as the one-parameter family of binary symmetric correlations (denoted $\msf{BSC}_{q}$) and binary erasure correlations (denoted $\msf{BEC}_{q}$), which are set by $q\in[0,1]$. We also compare these results to the case of symmetric key (denoted $\msf{SK}$), where perfect privacy is trivial. In Tab.~\ref{tab:privacy-errors} we report these results, and reference the relevant theorems where each optimal privacy error is proved.

On the one hand, these results, particularly in the case of QSNIR privacy, demonstrate operational leakages for primitive correlations that hold through non-interactive composition. Proving these exact errors necessitates a variety of technical insights into symmetries that emerge in the coherent embeddings of these correlations. Of note is that the value for $\delta_{\textup{QSNIR}}^{\msf{OK}}$ (which is the same given either a corrupt $\mbf{A}$ or $\mbf{B}$) is related to the golden ratio. This arises from how the proof for this error exploits a cyclic invariance in the bipartite connecting graph for this correlation.

On the other hand, the whole of Tab.~\ref{tab:privacy-errors} demonstrates a richness in the gap between $\delta_{\textup{MED}}$ and $\delta_{\textup{QSNIR}}$, which necessitates further exploration. We see examples where $\delta_{\textup{MED}}=\delta_{\textup{QSNIR}}=0$, where $\delta_{\textup{MED}}=\delta_{\textup{QSNIR}}>0$, where $\delta_{\textup{MED}}\neq\delta_{\textup{QSNIR}}$ with $\delta_{\textup{MED}}=0$ and $\delta_{\textup{QSNIR}}>0$, and where $\delta_{\textup{MED}}\neq\delta_{\textup{QSNIR}}$ and both are nonzero. This motivates future work to better understand the source of this gap, in relation to the operational decision-theoretic problems that MED and QSNIR address.

\newlength{\colw}
\setlength{\colw}{2.0cm}

\begin{table}[t]
    \centering
    \renewcommand{\arraystretch}{1.5}
    \begin{tabular}{|l||c|c||c|c|}
        \hline
        \multirow{2}{*}{$(\X,\Y)$} & \multicolumn{2}{c||}{$\delta_{\textup{MED}}$} & \multicolumn{2}{c|}{$\delta_{\textup{QSNIR}}$} \\
        \cline{2-5}
        & \makebox[\colw]{$\mbf{A}$} & \makebox[\colw]{$\mbf{B}$} & \makebox[\colw]{$\mbf{A}$} & \makebox[\colw]{$\mbf{B}$} \\
        \hline\hline
        $\msf{SK}$ (Cor.~\ref{cor:SK-priv})  & \multicolumn{2}{c||}{\makebox[2\colw]{$0$}} & \multicolumn{2}{c|}{\makebox[2\colw]{$0$}} \\
        \hline
        $\msf{OK}$ (Thm.~\ref{thm:priv-OK}) & \multicolumn{2}{c||}{\makebox[2\colw]{$(2\sqrt{2}-1)/8$}} & \multicolumn{2}{c|}{\makebox[2\colw]{$(\sqrt{5}-1)/4$}} \\
        \hline\hline
        $\msf{BSC}_{q}$ (Thm.~\ref{thm:BSC-priv}) & \multicolumn{2}{c||}{\makebox[2\colw]{$0$}} & \multicolumn{2}{c|}{\makebox[2\colw]{$\delta^{\msf{BSC}}_{\textup{QSNIR}}(q)$}} \\
        \hline
        $\msf{BEC}_{q}$ (Thm.~\ref{thm:BEC-priv}) & \makebox[\colw]{$\delta^{\msf{BEC}}_{\textup{MED},\mbf{A}}(q)$} & \makebox[\colw]{$\delta^{\msf{BEC}}_{\mbf{B}}(q)$} & \makebox[\colw]{$\delta^{\msf{BEC}}_{\textup{QSNIR},\mbf{A}}(q)$} & \makebox[\colw]{$\delta^{\msf{BEC}}_{\mbf{B}}(q)$} \\
        \hline
    \end{tabular}
    \setlength{\abovecaptionskip}{10pt}
    \caption{Optimal privacy errors for 2PC correlations. Cells are merged horizontally when the error is the same for both $\mbf{A}$ and $\mbf{B}$. For $\msf{BEC}_{q}$ the error is the same for MED and QSNIR privacy, when $\mbf{B}$ is corrupt. We therefore denote it by a single label.}
    \label{tab:privacy-errors}
\end{table}

\subsection{Beyond Canonical Embeddings}
The results in Tab.~\ref{tab:privacy-errors} are fixed for coherent embeddings, where there are explicitly no relative phases between the basis states supported on the correlation. Including these phases in our optimization problem renders it no longer convex. While we leave this larger optimization problem for future study, we do report two results, supporting the belief that these added phases can only increase the optimal QSNIR privacy error. For the case when $(\X,\Y)$ is a two-bit correlation, such that its embedding is a two-qubit state, we prove that $\delta_{\textup{QSNIR}}$ is upper-bounded by a quantity that is minimized when the relative phase is zero. We also report that our QSNIR privacy errors are numerically minimized when the phases are null, for several of the correlations we consider. These results are deferred to App.~\ref{app:BEY}.

\section{\label{sec:PRE}Preliminaries}
\subsection{\label{sec:PRE-corr}Classical Correlations}
Let $\Z$ denote a random variable ranging over the finite alphabet $\mc{Z}$, with a probability mass function (pmf) $p_{\Z}(z)\coloneqq\Pr[\Z=z]$ for $z\in\mc{Z}$.  We denote a sampling from $\Z$ as $z\samp \msf{Z}$. The Shannon entropy contained in $\Z$ is expressed as
\begin{equation}
    H(\Z)\coloneqq-\sum_{z\in\mc{Z}}p_\Z(z)\log p_\Z(z).
\end{equation}
When $\mc{Z}=\{0,1\}$, we can always choose $p_{\Z}(0)=1-q$ and $p_{\Z}(1)=q$, for some $q\in[0,1]$, and define the binary entropy function as $h(q)\coloneqq -q\log q-(1-q)\log(1-q)$.

Given two variables, $\Z$ and $\Z'$, taken over the same alphabet, their total variation distance is given by
\begin{equation}
    \tv(\Z,\Z')\coloneqq\frac{1}{2}\sum_{z\in\mc{Z}}\left|p_{\Z}(z)-p_{\Z'}(z)\right|, \label{eq:SD}
\end{equation}
where we say that $\Z=\Z'$ if and only if $\tv(\Z,\Z')=0$. Similarly, we define the Bhattacharyya coefficient between $\Z$ and $\Z'$ as
\begin{equation}
    \bc(\Z,\Z')\coloneqq\sum_{z\in\mc{Z}}\sqrt{p_{\Z}(z)p_{\Z'}(z)},
\end{equation}
where now $\Z=\Z'$ if and only if $\bc(\Z,\Z')=1$. In either case, these measures are symmetric and equality occurs when both random variables have identical pmfs.

Classical processing of the random variable $\Z$ into another, $\Z'$, over a possibly distinct alphabet, $\mc{Z}'$, refers to any stochastic map (i.e. a channel) $\mrm{N}\colon\mc{Z}\to\mc{Z}'$ that transforms $\Z\mapsto\mrm{N}(\Z)=\Z'$. For samples $z\samp\Z$ and $z'\samp\Z'$, the channel $\mrm{N}$ is well defined by the conditional pmf, $p_{\mrm{N}}(z'|z)\coloneqq\Pr[\Z'=z'|\Z=z]$. 

Let $\X$ and $\Y$ be random variables over respective alphabets $\mc{X}$ and $\mc{Y}$. We denote by $(\X,\Y)$ a pair of random variables composing a \textit{correlation}, over $\mc{X}\times\mc{Y}$, governed by a parent pmf $p_{\X\Y}(x,y)\coloneqq\Pr[\X=x,\Y=y]$. From this parent, the marginal distributions of $\X$ and $\Y$ are computed as $p_{\X}(x)=\sum_{y}p_{\X\Y}(x,y)$ and $p_{\Y}(y)=\sum_{x}p_{\X\Y}(x,y)$, respectively. Such a correlation is product when $p_{\X\Y}(x,y)=p_{\X}(x)p_{\Y}(y)$. Given a sample $x\samp \X$, we define a conditional random variable, $\Y|\X=x$, according to the conditional pmf $p_{\Y|\X}(y|x)=p_{\X\Y}(x,y)/p_{\X}(x)$. The conditional variable $\X|\Y=y$ is defined accordingly. Note that when discussing a single correlation, $(\X,\Y)$, we will often drop the subscript denoting the correlation by renaming $p\coloneqq p_{\X\Y}$, and likewise for its marginals and conditionals.

The joint entropy of a correlation is denoted by $H(\X,\Y)$. The entropy in $\Y$ conditioned on $\X$, $H(\Y|\X)$, is the average value of the entropy in $\Y|\X=x$, which is equivalent to the difference in entropy between $(\X,\Y)$ and $\X$ (i.e. $H(\Y|\X)=H(\X,\Y)-H(\X)$), and is similarly defined for $H(\X|\Y)$. The mutual information within a correlation, $I(\X\colon\Y)$, is a symmetric quantity associated with the amount of information that $\Y$ carries over $\X$, and vice versa, with asymptotically vanishing error.\footnote{Formally the mutual information is the Kullback-Leibler divergence between $(\X,\Y)$ and the product of its marginals.  It can be computed from the joint and marginal entropy as $I(\X\colon\Y)=H(\X)+H(\Y)-H(\X,\Y)$, or from the marginal and conditional entropy as $I(\X\colon\Y)=H(\X)-H(\X|\Y)=H(\Y)-H(\Y|\X)$.}

Lastly, the minimum-entropy of $\Y$ given $\X$, denoted $H_{\min}(\Y|\X)$, directly quantifies the one-shot minimum-error probability for guessing $y$ given $x$, and is computed as $H_{\min}(\Y|\X)=-\log\sum_{x}\max_{y}p(x,y)$. We also define
\begin{equation}
    P_{\textup{guess}}(\Y|\X)\coloneqq\sum_{x}\max_{y}p(x,y),
\end{equation}
as the optimal guessing probability. This definition is based on how the optimal single-copy strategy for guessing $y$ given $x$, is to simply guess the value that maximizes $p(x,y)$~\cite{Chor-1988-UB,Nisan-1996-RL}. $H_{\min}(\X|\Y)$ is defined and computed accordingly. 

\subsection{\label{sec:PRE-com-dep}Support Graphs, Common and Dependent Random Variables}
We next define ways to deterministically construct local random variables from $(\X,\Y)$, whose entropies quantify one-shot \textit{zero-error} information-theoretic properties of the underlying correlations~\cite{Wolf-2004-Zi}. These entropies serve as useful operational metrics for what is and isn't private given only a marginal of the correlation.

In order to define the common random variable of a correlation, we first define a graphical representation of the correlation.
\begin{definition}[Support Graph of $(\X,\Y)$]\label{def:supp-g}
    Let $(\X,\Y)$ be a correlation, over $\mc{X}\times\mc{Y}$, fixed by $p(x,y)$. Its support graph, $G_{\X\Y}=(V,E)$, is a bipartite graph, defined such that $V=\mc{X}\cup\mc{Y}$, and $ E=\{(x,y)\colon x\in\mc{X}, y\in\mc{Y},p(x,y)>0\}$.
\end{definition}
\noindent
From the support graph, we construct the common random variable as follows.
\begin{definition}[Common Random Variable ($\X\wedge\Y$)~\cite{Wolf-2004-Zi}]
    For $(\X,\Y)$, its common random variable, $\X\wedge\Y$, is the largest random variable that can be constructed deterministically from either $\X$ or $\Y$. Explicitly, $\X\wedge\Y\coloneqq f_{\X}(\X)=h_{\Y}(\Y)$, where
    \begin{itemize}
        \item $f_{\X}\colon\mc{X}\to2^{\mc{X}\cup\mc{Y}}$, such that $f_{\X}(x)$ is the set of $v\in V$ that  are in the connected component\footnote{The \textit{connected component} of a graph, $G$, with respect to a vertex $v\in V$, is the set of all vertices in the vertex-disjoint subgraph of $G$ that contains $v$.} of $G$, containing $x\in\mc{X}\subset V$.
        \item $h_{\Y}\colon\mc{Y}\to2^{\mc{X}\cup\mc{Y}}$, such that $h_{\Y}(y)$ is the set of $v\in V$ that are in the connected component of $G$, containing $y\in\mc{Y}\subset V$.
    \end{itemize}
\end{definition}
Operationally, $\X\wedge\Y$ is related to a one-shot local data compression problem. While $I(\X\colon\Y)$ quantifies the amount of common randomness obtainable from $(\X,\Y)$ per copy in the asymptotic limit, $H(\X\wedge\Y)$ (also referred to as the zero-error mutual information\footnote{This is also equivalent to the G\'{a}cs-K\"{o}rner common information~\cite{Gacs-1973-Ci}.}) quantifies the exact amount of common randomness obtained from each individual sample of the correlation. 

In contrast to the common random variable, we also define the dependent part of one random variable on another.
\begin{definition}[Dependent Random Variable ($\X\searrow\Y$)~\cite{Wolf-2004-Zi}]
    For $(\X,\Y)$, the dependent part of $\X$ with respect to $\Y$, $\X\searrow\Y$, is the smallest random variable that can be constructed deterministically from $\X$, such that $\X\longleftrightarrow \X\searrow\Y\longleftrightarrow\Y$ forms a Markov chain. Explicitly, $\X\searrow\Y\coloneqq  f_{\X}(\X)$, where $f_{\X}(x)=\Y|\X=x$ is a distribution.
\end{definition}
\noindent
The entropic quantity of interest for the dependent random variable is $H(\X\searrow\Y|\Y)$ (and the corresponding quantity with reversed labels). Unlike the common random variable, the operational interpretation of this quantity can be more opaque. In some sense, $\X\searrow\Y$ is the part of $\X$ that is correlated with $\Y$, after we have removed all the parts of $\X$ which are explicitly independent.\footnote{One way to see this is from the identity $I(\X\colon\Y)=H(\X\searrow\Y)-H(\X\searrow\Y|\Y)$ (Cor. 2 in~\cite{Wolf-2004-Zi}). $H(\X\searrow\Y)$ therefore is the mutual information if and only if there is no uncertainty in $\X\searrow\Y$ given $\Y$, i.e. that $H(\X\searrow\Y|\Y)=0$.}  
Of note is that when $H(\X\searrow\Y|\Y)=0$, for any $(x,x')$ which are supported on the correlation with \textit{any} $y$, we must then have that $\bc(\Y|\X=x,\Y|\X=x')=1$. Another way to say this is that, for any $(\X,\Y)$ satisfying $H(\X\searrow\Y|\Y)=0$, the conditional $\Y|\X=x$ is either a point distribution, or the same distribution for all $x$ that are supported with more than a single $y$. 

These one-shot entropies are useful in the context of non-interactive reductions, as they satisfy data processing inequalities.
\begin{lemma}[Data Processing for Zero-Error Entropies~\cite{Wolf-2005-NM}]
    Let $(\X,\Y)$ be a correlation, and let $(\mrm{M},\mrm{N})$ be local channels. The following data processing inequalities must hold:
    \begin{subequations}
        \begin{align}
            I(\X\colon\Y)-H(\X\wedge\Y) &\geq I(\mrm{M}(\X)\colon\mrm{N}(\Y))- H(\mrm{M}(\X)\wedge\mrm{N}(\Y)), \\
            H(\X\searrow\Y|\Y) & \geq H(\mrm{M}(\X)\searrow\mrm{N}(\Y)|\mrm{N}(\Y)), \\
            H(\Y\searrow\X|\X) & \geq H(\mrm{N}(\Y)\searrow\mrm{M}(\X)|\mrm{M}(\X)).
        \end{align}
    \end{subequations}
\end{lemma}
\noindent
In addition to being nonincreasing under local stochastic processing, these entropies are also nonincreasing under lossless communication, such that they are general monotones for privacy in 2PC, beyond the non-interactive setting. Consequently, $H(\X\searrow\Y|\Y)=0$ if and only if $H(\Y\searrow\X|\X)=0$, and if and only if $ I(\X\colon\Y)=H(\X\wedge\Y)$. When $H(\X\searrow\Y|\Y)=0$ holds, we say the correlation is \textit{trivial}. Conversely, when $H(\X\wedge\Y)$ is large enough, any correlation can be constructed non-interactively~\cite{Gacs-1973-Ci,Witsenhausen-1975-SP,Wyner-1975-ci}.

\subsection{\label{sec:PRE-SNIR}Secure Non-Interactive Reductions (SNIR)}
Formally, SNIR exists within the UC model~\cite{Canetti-2001-Uc}, simplified for non-interactive, input-independent functionalities. In this restricted setting, UC security coincides with the simpler standalone model for simulation-based security~\cite{Kushilevitz-2009-IS}. Here, there is no distinction between malicious and honest-but-curious security, as well as static and adaptive corruption. Hence, a SNIR requires demonstrating only a single straight-line (non-rewinding) simulator for each of $\mbf{A}$ and $\mbf{B}$ to prove privacy.

Intuitively, the goal of a SNIR is to demonstrate that the adversary's marginal view of a \textit{source} correlation reveals nothing more than the correct marginal view of the \textit{target} correlation. This is proved by demonstrating that a simulator can reconstruct this view of the \textit{real world} source correlation, given only the corresponding view of the \textit{ideal world} target correlation. Let $(\X_{\src},\Y_{\src})$ and $(\X_{\tgt},\Y_{\tgt})$ be the source and target correlation, respectively. 
\begin{definition}[$(\epsilon,\delta_{\mbf{A}},\delta_{\mbf{B}})$-SNIR]\label{def:eSNIR}
    A pair of local channels $(\mrm{N}^{\mbf{A}},\mrm{N}^{\mbf{B}})$, where $\mrm{N}^{\mbf{A}}\colon\mc{X}_{\src}\to\mc{X}_{\tgt}$ and $\mrm{N}^{\mbf{B}}\colon\mc{Y}_{\src}\to\mc{Y}_{\tgt}$, forms an $(\epsilon,\delta_{\mbf{A}},\delta_{\mbf{B}})$-SNIR from $(\X_{\tgt},\Y_{\tgt})$ to $(\X_{\src}, \Y_{\src})$ if the following conditions are satisfied:
    \begin{enumerate}
        \item[\textup{1.}] \textup{\textbf{Correctness}}. When both parties are honest,
        \begin{equation}
            \tv\left(\left(\mrm{N}^{\mbf{A}}\left(\X_{\src}\right),\mrm{N}^{\mbf{B}}\left(\Y_{\src}\right)\right),\left(\X_{\tgt},\Y_{\tgt}\right)\right)\leq\epsilon. \label{eq:snir-cor-static}
        \end{equation}
        \item[\textup{2.}] \textup{\textbf{Simulation-based privacy against $\mbf{A}$.}} When $\mbf{B}$ is honest, there exists a simulator, given by the channel $\mrm{S}^{\mbf{A}}\colon\mc{X}_{\tgt}\to \mc{X}_{\src}$, satisfying
        \begin{equation}
           \tv\left(\left(\X_{\src},\mrm{N}^{\mbf{B}}\left(\Y_{\src}\right)\right),\left(\mrm{S}^{\mbf{A}}\left(\X_{\tgt}\right),\Y_{\tgt}\right)\right)\leq\delta_{\mbf{A}}.\label{eq:snir-priv-static-A}
        \end{equation}
        \item[\textup{3.}] \textup{\textbf{Simulation-based privacy against $\mbf{B}$.}} When $\mbf{A}$ is honest, there exists a simulator, given by the channel $\mrm{S}^{\mbf{B}}\colon\mc{Y}_{\tgt}\to \mc{Y}_{\src}$, satisfying
        \begin{equation}
            \tv\left(\left(\mrm{N}^{\mbf{A}}\left(\X_{\src}\right),\Y_{\src}\right),\left(\X_{\tgt},\mrm{S}^{\mbf{B}}\left(\Y_{\tgt}\right)\right)\right)\leq\delta_{\mbf{B}}.\label{eq:snir-priv-static-B}
        \end{equation}
    \end{enumerate}
\end{definition}
\noindent When $\epsilon=0$ we say the reduction has \textit{perfect correctness}. When $ \delta_{\mbf{A}}=\delta_{\mbf{B}}\eqqcolon\delta=0$, we say that the reduction has \textit{perfect privacy}. When both $\epsilon=0$ and $\delta=0$, we say the reduction is perfect, or express that it is a $0$-SNIR.

In~\cite{Agarwal-2022-SN}, the authors explore a variety of criteria for the existence of two-party SNIRs, analyzing the spectral properties of correlations and reductions between them.\footnote{For example, $0$-SNIRs exist between correlations with uniform marginals, only when the eigenvalues of the classical Gram matrix, with respect to either side of the distribution, for the target are contained within the source.}
Within the present scope, we focus specifically on extracting correlations from entanglement. The coherent nature of how these correlations are embedded in quantum systems admits new nontrivialities when deriving bounds on security, such that we do not expect results of classical SNIRs to directly carry over. 

\subsection{\label{sec:PRE-quant}Quantum Resources and Coherent Embeddings} 
The physical states of a quantum system, $\rho$, are represented by trace-one, positive-semidefinite operators, known as density matrices, acting on a complex Hilbert space $\mc{H}$. We let $\D(\mc{H})=\{\rho:\rho\succeq0,\tr[\rho]=1\}$ denote the set of all density matrices on $\mc{H}$. Pure states are rank-one density matrices, i.e. $\rho=\op{\psi}{\psi}$, which are often more simply expressed by normalized vectors $\ket{\psi}\in\mc{H}$.\footnote{At times we will refer to a pure density matrix simply as $\psi$.}  In this work, we only consider finite-dimensional Hilbert spaces. It is also conceptually helpful to represent a classical random variable $\Z$ as a density matrix that is diagonal in the computational basis, with eigenvalues equal to the pmf fixing $\Z$. We write $\rho\sim\Z$ when this correspondence holds:
\begin{align}
    \rho\sim\Z\quad\Leftrightarrow\quad \rho=\sum_{z\in\mc{Z}}p_{}(z)\op{z}{z}.
\end{align}

We define symmetric distances between quantum states as follows.
The trace distance between two states $\rho,\sigma\in\D(\mc{H})$ is defined as
\begin{equation}
    \td\left(\rho,\sigma\right)\coloneqq\frac{1}{2}\left\lVert\rho-\sigma\right\rVert_{1}=\frac{1}{2}\tr\left[\left|\rho-\sigma\right|\right]. 
\end{equation}
Again we say that $\rho=\sigma$ if and only if $\td(\rho,\sigma)=0$. The fidelity between $\rho$ and $\sigma$ is defined as
\begin{equation}
    \f(\rho,\sigma)\coloneqq\left\|\sqrt{\rho}\sqrt{\sigma}\right\|_{1}=\tr\left[\sqrt{\sqrt{\rho}\sigma\sqrt{\rho}}\right],
\end{equation}
where $\rho=\sigma$ if and only if $\f(\rho,\sigma)=1$. Note that $\td(\rho,\sigma)=\tv(\Z,\Z')$ and  $\f(\rho,\sigma)=\bc(\Z,\Z')$, whenever $\rho\sim\Z$ and $\sigma\sim\Z'$.

Physical transformations performed on a quantum system are represented by completely-positive trace-preserving (CPTP) maps, also known as quantum channels.  A CPTP map $\mc{N}\colon\D(\mc{H})\to \D(\mc{H}')$ sends every density matrix $\rho\in\D\left(\mc{H}\right)$ to another density matrix $\rho'=\mc{N}(\rho)\in\D(\mc{H}')$.  We let $\id$ denote the trivial map on $\D\left(\mc{H}\right)$ such that $\id(\rho)=\rho $ for all $\rho\in\D\left(\mc{H}\right)$. As with the case of states, it is useful to similarly represent a classical channel $\textup{N}\colon\mc{Z}\to\mc{Z}'$ as a CPTP map $\mc{N}\colon\D(\mbb{C}^{|\mc{Z}|})\to \D(\mbb{C}^{|\mc{Z}'|})$. We write $\mc{N}\sim \textup{N}$, such that for some $\rho\sim\Z$, we have the correspondence:
\begin{equation}
    \mc{N}\sim \textup{N} \quad\Leftrightarrow\quad  \mc{N}(\rho) = \sum_{z'\in\mc{Z}'}\sum_{z\in\mc{Z}}p_{\mrm{N}}\left(z'\middle|z\right)\op{z'}{z}\rho\op{z}{z'},
\end{equation}
which outputs a state $\rho'\sim\Z'$. 

A special class of CPTP maps are \textit{measurement maps}, $\mc{M}\colon\D(\mc{H})\to\D(\mbb{C}^{|\mc{Z}|})$, which have the form
\begin{align}
   \mc{M}\colon\quad \rho \mapsto \sum_{z\in\mc{Z}} \tr[M_{z}\rho]\op{z}{z},
\end{align}
where the $\{M_{z}\}_{z\in\mc{Z}}$ form a positive operator-valued measure (POVM), meaning that the $M_{z}\succeq0$ for all $z$ and $\sum_{z}M_{z}=\mbb{I}$.  When the POVM elements are furthermore orthogonal, satisfying $M_{z}M_{z'}=\delta_{z,z'}M_{z}$ for any pair of $z,z'\in\mc{Z}$ (where $\delta_{z,z'}$ is the Kronecker delta), we further specify the POVM as a projection-valued measure (PVM), and each element as an orthogonal projector. From Born's rule, the POVM associates with every quantum state $\rho$ a pmf for $z\in\mc{Z}$ given by $p(z)=\tr[M_{z}\rho]$.  We can thus interpret every measurement map $\mc{M}$ as transforming a given $\rho$ into the random variable $\Z$, ranging over $\mc{Z}$, for which $\mc{M}(\rho)\sim\Z$.

The opposite of a measurement map is a preparation map $\mc{P}\colon\D\left(\mbb{C}^{|\mc{Z}|}\right)\to\D\left(\mc{H}\right)$, which has the form
\begin{align}
    \mc{P}\left(\op{z}{z}\right)=\sigma_z,\qquad z\in\mc{Z};\;\; \sigma_z\in\D\left(\mc{H}\right).
\end{align}
Intuitively, the map $\mc{P}$ prepares a state, $\sigma_z$, whenever it is fed a sample $z\samp\Z$.

A bipartite quantum system, for parties $\{\mbf{A},\mbf{B}\}$, is represented within a tensor product space $\mc{H}^{\mbf{AB}}\coloneqq\mc{H}^{\mbf{A}}\otimes\mc{H^{\mbf{B}}}$.  
If $\rho^{\mbf{AB}}$ is the global state shared by the parties, then the state of an individual party, say $\mbf{A}$, is obtained by taking a partial trace over the other system, $\rho^{\mbf{A}}=\tr_{\mbf{B}}[\rho^{\mbf{AB}}]$. 
A bipartite state is product, when the joint state is the tensor product of its marginals, i.e. $\rho^{\mbf{AB}}=\rho^{\mbf{A}}\otimes\rho^{\mbf{B}}$. 

We say that a bipartite CPTP map, $\mc{N}^\mbf{AB}$, is a \textit{local map} if it decomposes into a tensor product of CPTP maps, $\mc{N}^\mbf{AB}=\mc{N}^{\mbf{A}}\otimes\mc{N}^{\mbf{B}}$. A \textit{local measurement} similarly has the form $\mc{M}^\mbf{AB}=\mc{M}^{\mbf{A}}\otimes\mc{M}^{\mbf{B}}$, defined by a pair of POVMs, $\{M_{x}^{\mbf{A}}\}_{x\in\mc{X}}$ and $\{M_{y}^{\mbf{B}}\}_{y\in\mc{Y}}$.
For the input state $\rho^{\mbf{AB}}$, we thus have $\mc{M}^\mbf{AB}(\rho^{\mbf{AB}})\sim (\X,\Y)$, where $ (\X,\Y)$ is a correlation with parent pmf
\begin{align}
    p\left(x,y\right)=\tr\left[\left(M_{x}^{\mbf{A}}\otimes M_{y}^{\mbf{B}}\right)\rho^{\mbf{AB}}\right]
\end{align}
over the alphabet $\mc{X}\times\mc{Y}$.

Local maps play a central role in this work.
Of particular interest are local maps in which $\mbf{A}$ or $\mbf{B}$ acts trivially. For example, consider a pair of measurement channels, $\mc{M}^{\mbf{A}}\otimes\id^{\mbf{B}}$. $\mbf{A}$'s measurement channel will transform $\rho^{\mbf{AB}}\in\D(\mc{H}^{\mbf{AB}})$ into a classical-quantum (cq) state
\begin{equation}
    \rho^{\X\mbf{B}}=\mc{M}^{\mbf{A}}\otimes\id^{\mbf{B}}(\rho^{\mbf{AB}})=\sum_{x}\op{x}{x}^{\X}\otimes \tr_{\mbf{A}}\left[\left(M_{x}^{\mbf{A}}\otimes\mbb{I}^{\mbf{B}}\right)\rho^{\mbf{AB}}\right],
\end{equation}
where to be explicit we write that $\rho^{\X\mbf{B}}\in\D(\mbb{C}^{|\mc{X}|}\otimes\mc{H}^{\mbf{B}})$. Here, $\mbf{B}$ can still hold a coherent quantum system, while $\mbf{A}$ has measured their system and just holds the outcome $x\samp\X$, which is denoted as a classical register carrying the superscript $\X$ in the cq state. Along these lines, a correlation $(\X,\Y)$, initially shared between $\{\mbf{A},\mbf{B}\}$, is transformed into a cq state by applying a preparation map to one of the subsystems in the state $\rho^{\X\Y}\sim(\X,\Y)$, e.g. $\id^{\mbf{A}}\otimes\mc{P}^{\mbf{B}}(\rho^{\X\Y})=\rho^{\X\mbf{B}}$. 

While we will not make use of the von Neumann entropies that are defined for quantum systems, we will refer to the quantum minimum-entropy, denoted $H_{\min}(\mbf{A}|\mbf{B})_{\rho}$ for $\rho^{\mbf{AB}}$. Formally,
\begin{equation}
    H_{\min}(\mbf{A}|\mbf{B})_{\rho}:=\sup_{\sigma^{\mbf{B}},\lambda}\{\lambda\in\mbb{R}, \sigma_{\mbf{B}}\in \D(\mc{H}^{\mbf{B}}):\rho^{\mbf{AB}}\succeq 2^{-\lambda}\mathbb{I}^{\mbf{A}}\otimes\sigma^{\mbf{B}}\},
\end{equation}
however, in the context of cq states, $\rho^{\X\mbf{B}}$, it picks up an operational interpretation as the ability of $\mbf{B}$ to guess an $x\samp\X$, given their quantum side-information~\cite{Konig-2009-OM}. 

\subsubsection{Coherent Embeddings.}  We define a special class of states which encodes a fixed correlation.
\begin{definition}[Embeddings]\label{def:coh-embed}
    Let $(\X,\Y)$ be a correlation, fixed by a parent pmf $p(x,y)$, and let $\varphi(x,y)$ be a real function, where  $\varphi:\mc{X}\times\mc{Y}\to[0,2\pi)$.
    We define the coherent embedding of $(\X,\Y)$ as the pure state
    \begin{equation}
        \ket{\mrm{\Psi}_{\X\Y}}^{\mbf{AB}}\coloneqq\sum_{x,y}\sqrt{p(x,y)}e^{i\varphi(x,y)}\ket{x}^{\mbf{A}}\otimes\ket{{y}}^{\mbf{B}}.
    \end{equation}
    We refer to $\varphi(x,y)$ as the additional \textup{phase information} encoded in the embedding.
    When $\varphi(x,y)=0$, for all $x$ and $y$, we say that $\ket{\mrm{\Psi}_{\X\Y}}^{\mbf{AB}}$ is the \textup{canonical} embedding. Note that unless $(\X,\Y)$ is a product distribution, $\ket{\mrm{\Psi}_{\X\Y}}^{\mbf{AB}}$ is necessarily entangled. 
\end{definition}
\noindent
Clearly given such a global state, the local measurement channels $\mc{M}^{\mbf{A}}\otimes\mc{M}^{\mbf{B}}$, given by PVMs $\{\op{x}{x}\}_{x}$ and $\{\op{y}{y}\}_{y}$, satisfy $p(x,y)=|\ip{x,y}{\mrm{\Psi}_{\X\Y}}|^{2}$.

With these states, we will be particularly interested in the reduced state of one party, given the computational basis measurement of the other. When $\mbf{A}$ measures their part of the embedding via the PVM, $\{\op{x}{x}\}_{x}$, this transforms the system into a cq state
\begin{equation}\label{eq:cq-embedding}
    \mc{M}^{\mbf{A}}\otimes\id^{\mbf{B}}(\op{\mrm{\Psi}_{\X\Y}}{\mrm{\Psi}_{\X\Y}}^{\mbf{AB}})=\sum_{x}\op{x}{x}^{\X}\otimes p(x)\op{\psi_{x}}{\psi_{x}}^{\mbf{B}},
\end{equation}
where
\begin{equation}
    \ket{\psi_{x}}\coloneqq \sum_{y}\sqrt{p(y|x)}e^{i\varphi(x,y)}\ket{y}
\end{equation}
is the reduced state of $\mbf{B}$'s system, conditioned on $x\samp\X$. 

\subsubsection{Minimum-Error State Discrimination (MED).}
Let $\{\rho_{x}\}_{x\in\mc{X}}$ be an ensemble of states, with prior probabilities $\{p(x)\}_{x\in\mc{X}}$, sourced by a classical random variable $\X$. As its name suggests, the goal of MED is to discern a particular $x\samp \X$, with minimum-error, by making a single measurement on a single copy of $\rho_{x}$. The optimal success probability of MED serves as a one-shot estimate for learning $x$, given some amount of quantum side information. We later use this success probability to construct a one-shot privacy definition. We reference~\cite{Barnett-2009-Qs} for a larger survey on state discrimination. In App.~\ref{app:med} we detail the explicit MED results and lemmas that appear in our proofs.

\section{\label{sec:QSNIR}Quantum Secure Non-Interactive Reductions (QSNIR)}
The definition we give for QSNIR handles non-interactive reductions between quantum states in general, beyond the case when the target state is a classical correlation. We allow the adversary arbitrary computational power (both quantum and classical) in their attack, giving them access to unlimited quantum side information, which they can entangle with their local registers of the initial resource state. As this framework is non-interactive, these malicious actions will appear indistinguishable to the honest party, from when the adversary is passive. Hence, the adversary in this setting is equivalent to the \textit{quantum honest-but-curious} setting of \cite{Salvail-2009-PT}.
\begin{definition}[$(\epsilon,\delta_{\mbf{A}},\delta_{\mbf{B}})$-QSNIR]\label{def:eQSNIR}
    A pair of local channels, $\mc{N}^{\mbf{AB}}\coloneqq\mc{N}^{\mbf{A}}\otimes\mc{N}^{\mbf{B}}$, where $\mc{N}^{\mbf{A}}\colon\D(\mc{H}_{\src}^{\mbf{A}})\to\D(\mc{H}_{\tgt}^{\mbf{A}})$ and $\mc{N}^{\mbf{B}}\colon\D(\mc{H}_{\src}^{\mbf{B}})\to\D(\mc{H}_{\tgt}^{\mbf{B}})$, forms an $(\epsilon,\delta_{\mbf{A}},\delta_{\mbf{B}})$-QSNIR from $\rho^{\mbf{AB}}_{\textup{tgt}}$ to $\rho^{\mbf{AB}}_{\textup{src}}$ if the following conditions are satisfied:
    \begin{enumerate}
        \item[\textup{1.}] \textup{\textbf{Correctness}}. When both parties are honest,
        \begin{equation}
            \td\left(\left(\mc{N}^{\mbf{A}}\otimes\mc{N}^{\mbf{B}}\right)\left(\rho^{\mbf{AB}}_{\textup{src}}\right),\rho^{\mbf{AB}}_{\textup{tgt}}\right)\leq\epsilon. \label{eq:qsnir-cor-static}
        \end{equation}
        \item[\textup{2.}] \textup{\textbf{Simulation-based privacy against $\mbf{A}$.}} When $\mbf{B}$ is honest, there exists a simulator, given by the channel $\mc{S}^{\mbf{A}}\colon\D(\mc{H}_{\tgt}^{\mbf{A}})\to \D(\mc{H}_{\src}^{\mbf{A}})$, satisfying
        \begin{equation}
            \td\left(\left(\id^{\mbf{A}}\otimes \mc{N}^{\mbf{B}}\right)\left(\rho^{\mbf{AB}}_{\textup{src}}\right), \left(\mc{S}^{\mbf{A}}\otimes \id^{\mbf{B}}\right)\left(\rho^{\mbf{AB}}_{\textup{tgt}}\right)\right)\leq \delta_{\mbf{A}}.\label{eq:qsnir-priv-static-A}
        \end{equation}
        \item[\textup{3.}] \textup{\textbf{Simulation-based privacy against $\mbf{B}$.}} When $\mbf{A}$ is honest, there exists a simulator, given by the channel $\mc{S}^{\mbf{B}}\colon\D(\mc{H}_{\tgt}^{\mbf{B}})\to \D(\mc{H}_{\src}^{\mbf{B}})$, satisfying
        \begin{equation}
            \td\left(\left(\mc{N}^{\mbf{A}}\otimes \mrm{id}^{\mbf{B}}\right)\left(\rho^{\mbf{AB}}_{\textup{src}}\right), \left(\mrm{id}^{\mbf{A}}\otimes \mc{S}^{\mbf{B}}\right)\left(\rho^{\mbf{AB}}_{\textup{tgt}}\right)\right)\leq \delta_{\mbf{B}}.\label{eq:qsnir-sec-static}
        \end{equation}
    \end{enumerate}
\end{definition}
\noindent 
Just as with the classical definition, when $\epsilon=0$, the reduction has \textit{perfect correctness}, and when $\delta=0$, the reduction has \textit{perfect privacy}. When both $\epsilon=0$ and $\delta=0$, we express that the reduction is perfect, or that it is a $0$-QSNIR. Note that when we choose $\rho^{\msf{XY}}_{\tgt}\sim(\X_{\tgt},\Y_{\tgt})$ and $\rho^{\msf{XY}}_{\src}\sim(\X_{\src},\Y_{\src})$ as the target and source states, and $(\mc{N}^{\mbf{A}}\sim\mrm{N}^{\mbf{A}},\mc{N}^{\mbf{B}}\sim\mrm{N}^{\mbf{B}})$ as our pair of local channels, we exactly recover the definition for SNIR.

Generally, we prove that security is maintained through the parallel and sequential composition of two QSNIRs, with at worst additive error. We first bound the total error through parallel composition.
\begin{theorem}[Parallel Composition]\label{thm:parallel}
    Let $\mc{N}^{\mbf{AB}}_{1}$ be an $(\epsilon_{1},\delta_{\mbf{A},1},\delta_{\mbf{B},1})$-QSNIR from $\rho^{\mbf{AB}}_{1,\tgt}$ to $\rho^{\mbf{AB}}_{1,\src}$, and let $\mc{N}^{\mbf{AB}}_{2}$ be an $(\epsilon_{2},\delta_{\mbf{A},2},\delta_{\mbf{B},2})$-QSNIR from $\rho^{\mbf{AB}}_{2,\tgt}$ to $\rho^{\mbf{AB}}_{2,\src}$. Then, the local channel $\mc{N}^{\mbf{AB}}=\mc{N}^{\mbf{AB}}_{1}\otimes\mc{N}^{\mbf{AB}}_{2}$ satisfies an $(\epsilon_{1}+\epsilon_2,\delta_{\mbf{A},1}+\delta_{\mbf{A},2},\delta_{\mbf{B},1}+\delta_{\mbf{B},2})$-QSNIR, when reducing $\rho^{\mbf{AB}}_{1,\tgt}\otimes\rho^{\mbf{AB}}_{2,\tgt}$ to $\rho^{\mbf{AB}}_{1,\src}\otimes\rho^{\mbf{AB}}_{2,\src}$.
\end{theorem}
\noindent 
Similarly, we prove the error bound on the sequential composition of two QSNIRs.
\begin{theorem}[Sequential Composition]\label{thm:sequential}
      Let $\mc{N}^{\mbf{AB}}_{1}$ be an $(\epsilon_{1},\delta_{\mbf{A},1},\delta_{\mbf{B},1})$-QSNIR from $\rho^{\mbf{AB}}_{1,\tgt}$ to $\rho^{\mbf{AB}}_{0,\textup{src}}$, and let $\mc{N}^{\mbf{AB}}_{2}$ be an $(\epsilon_{2},\delta_{\mbf{A},2},\delta_{\mbf{B},2})$-QSNIR from $\rho^{\mbf{AB}}_{2,\tgt}$ to $\rho^{\mbf{AB}}_{1,\tgt}$. Then, the local channel $\mc{N}^{\mbf{AB}}=\mc{N}^{\mbf{AB}}_{2}\circ\mc{N}^{\mbf{AB}}_{1}$ satisfies an $(\epsilon_{1}+\epsilon_2,\delta_{\mbf{A},1}+\delta_{\mbf{A},2},\delta_{\mbf{B},1}+\delta_{\mbf{B},2})$-QSNIR, when reducing from $\rho^{\mbf{AB}}_{2,\tgt}$ to $\rho^{\mbf{AB}}_{0,\src}$.
\end{theorem}
\noindent
Proofs for these theorems can be found in App.~\ref{app:thm-para} and App.~\ref{app:thm-seq}, respectively.

Lastly, in App.~\ref{app:QSNIR-LUequi}, we prove that QSNIRs described by local unitary channels will always satisfy $\epsilon=\delta$, and when both quantities are $0$, we have \textit{local-unitary} equivalent resources.

\section{\label{sec:BND}Privacy Bounds for Reductions from Classical Correlations to Embeddings}
Our motivation for the QSNIR framework is the problem of disseminating correlated randomness for preprocessing in 2PC. Specifically, a QSNIR can demonstrate that a sufficient assumption for distributing a target correlation is to first distribute some source quantum state. Suppose $\rho_{\tgt}^{\X\Y}\sim(\X,\Y)$. Correctness is then related to how faithfully the correlation $(\X,\Y)$ can be extracted from the source quantum state, $\rho_{\src}^{\mbf{AB}}$, when both parties are honest. Each privacy error quantifies what a cheating party is able to learn about the other's honestly obtained samples, beyond what they would infer from their own ideal samples.

From a reduction that solely promises perfect correctness ($\epsilon=0$), what is the optimal $\delta$?
To answer this question we further restrict to a setting where the source state is a coherent embedding of $(\X,\Y)$, $\rho_{\textup{src}}^{\mbf{AB}}=\op{\mrm{\Psi_{\X\Y}}}{\mrm{\Psi}_{\X\Y}}^{\mbf{AB}}$, from Def.~\ref{def:coh-embed}. These states naturally admit reductions with perfect correctness. When honest, $\mbf{A}$ and $\mbf{B}$ measure in the computational basis and obtain samples $x\samp\X$ and $y\samp\Y$, respectively, recalling that $|\ip{x,y}{\mrm{\Psi}_{\X\Y}}|^{2}=p(x,y)$. 

Finally, in this section we explicitly fix $\varphi(x,y)=0$, such that we only consider reductions to canonical embeddings. As the goal is to make $\delta$ as small as possible, there may not seem any reason \textit{a priori} to assume that $\varphi(x,y)=0$ minimizes the privacy error for each party.\footnote{In~\cite{Salvail-2009-PT}, when defining the leakage of correlations with high enough dimension in terms of von Neumann entropies, there exist ensembles which are easier to distinguish pairwise, but have smaller entropy when the priors are chosen correctly. This result was first demonstrated in~\cite{Jozsa-2000-Ds}.} 
On the other hand, there is some intuition to believe that these phases can only increase $\delta$, as they are often free parameters that can only increase how well the adversary's pure state ensemble can be discriminated, increasing $\delta_{\textup{MED}}$.\footnote{For MED between pure state ensembles, the success probability generally depends on fewer parameters than the number of angles needed to fully specify each pure state in the ensemble within $\mbb{C}^{d}$~\cite{Barnett-2009-Qs,Bae-2013-Sm}.} We return to this problem in App.~\ref{app:BEY}.

For now, suppose the pair $\{\mbf{A},\mbf{B}\}$ possess the embedding $\ket{\mrm{\Psi}_{\X\Y}}^{\mbf{AB}}$, which they are guaranteed is pure, and suppose that $\mbf{A}$ is honest. As the correlation is arbitrary here, we will only focus on the case of a corrupt $\mbf{B}$. The optimal privacy error is given by
\begin{equation}
    \delta_{\mbf{B}} \coloneqq \min_{\mc{S}} \sum_{x} p(x)\textup{TD}\left(\op{\psi_{x}}{\psi_{x}}^{\mbf{B}},\mc{S}^{\mbf{B}}\left(\sum_{y}p(y|x)\op{y}{y}^{\Y}\right)\right), \label{eq:QSNIR-priv-from-embedding}
\end{equation}
where we have explicitly written out the reduced quantum state of $\mbf{B}$ in the real and ideal worlds. First note that 
\begin{equation}
    \mc{D}\left(\op{\psi_{x}}{\psi_{x}}\right)=\sum_{y}p(y|x)\op{y}{y},\label{eq:dephasing-rw-ens}
\end{equation}
where $\mc{D}$ is the completely dephasing channel, $\mc{D}(\rho)=\sum_{y}\ip{y}{\rho|y}\op{y}{y}$, which also describes the honest channel that $\mbf{B}$ would implement on average in the reduction. Therefore, QSNIR privacy in this context can be seen as a statement about how well we can invert this honest map. The linearity of $\mc{S}$ lets us phrase this as the following recovery problem: \textit{Given a coherent embedding, $\ket{\mrm{\Psi}_{\X\Y}}^{\mbf{AB}}$, and an honest $\mbf{A}$, find a simulator, which on average best inverts the projective measurement that an honest $\mbf{B}$ would implement}.

\subsection{\label{sec:BND-PERF}Only Trivial Correlations Admit Perfect Privacy}
We derive the necessary and sufficient conditions for $\delta=0$ privacy to exist, which demonstrates equivalently to~\cite{Salvail-2009-PT} that the correlation must be trivial. 
\begin{theorem}\label{thm:perf-priv}
    A reduction from $(\X,\Y)$ to its canonical embedding $\ket{\mrm{\Psi}_{\X\Y}}^{\mbf{AB}}$ has perfect privacy (against either a corrupt $\mbf{A}$ or $\mbf{B}$), if and only if $H(\Y\searrow\X|\X)=0$, or equivalently  $H(\X\searrow\Y|\Y)=0$, i.e. if and only if the correlation is trivial.
\end{theorem}
\begin{proof}
    From Eq.~\ref{eq:QSNIR-priv-from-embedding}, we see that $\delta=0$ implies that there exists an optimal simulator such that
    \begin{equation}\label{eq:perf-priv-avg}
        \sum_{y}p(y|x)\mc{S}_{\textup{opt}}\left(\op{y}{y}\right)= \op{\psi_{x}}{\psi_{x}}, \quad \forall x .
    \end{equation}
    When there exists a $y$ such that $p(y|x)>0$ and $p(y|x')>0$, we must then have that  $\ip{\psi_{x}}{\psi_{x'}}=1$, as a convex combination of states is pure if and only if every state in the ensemble is the same pure state.\footnote{This is exactly the proof of Cor. 3, in~\cite{Chefles-2004-ep}.} Generally,
    \begin{equation}
        \ip{\psi_{x}}{\psi_{x'}}=\sum_{y}\sqrt{p(y|x)p(y|x')}=\textup{BC}\left(\Y|\X=x,\Y|\X=x'\right),
    \end{equation}
    which is the Bhattacharyya coefficient between the conditionals. Therefore, for any pair of $x,x'$ which are in the support of $(\X,\Y)$, for the same $y$, the corresponding conditionals $\Y|\X=x$ and $\Y|\X=x'$ must be identical for all $y$,  which is equivalent to the condition that $H(\X\searrow\Y|\Y)=0$. 

    Now we prove the reverse direction, by demonstrating that if $H(\X\searrow\Y|\Y)=0$, we can always construct $\mc{S_{\textup{opt}}}$, explicitly. 
    Let $\mc{X}_{y}\coloneqq\{x\in\mc{X}\colon p(x,y)>0\}$, be the set of $x$ supported with $y$.
    Given a correlation $(\X,\Y)$, where $H(\X\searrow\Y|\Y)=0$, consider a simulator $\mc{S}_{\textup{opt}}$, such that
    \begin{equation}
        \mc{S}_{\textup{opt}}(\op{y}{y})=\op{\psi_{x}}{\psi_{x}},
    \end{equation}
    where $x$ is any element of $\mc{X}_{y}$. Suppose the simulator receives $y_{\textup{act}}\samp\Y$ and an honest $\mbf{A}$ receives $x_{\textup{act}}\samp\X$, denoting a fixed instance of ``actual'' samples from the ideal correlation. In the construction above, the simulator can arbitrarily choose any $x\in\mc{X}_{y_{\textup{act}}}$, and prepare $\ket{\psi_{x}}$ on $\mbf{B}$'s registers. Importantly, $x_{\textup{act}}\in\mc{X}_{y_{\textup{act}}}$, as well by the entropy constraint. Therefore, on average
    \begin{equation}
        \op{\psi_{x_{\textup{act}}}}{\psi_{x_\textup{act}}}=\sum_{y}p(y|x_{\textup{act}})\mc{S}_{\textup{opt}}(\op{y}{y})=\sum_{y}p(y|x)\op{\psi_{x}}{\psi_{x}}=\op{\psi_{x}}{\psi_{x}},
    \end{equation}
    as $p(y|x_{\textup{act}})=p(y|x)$ for all $y$.
    
    Lastly, $H(\X\searrow\Y|\Y)=0$ if and only if $H(\Y\searrow\X|\X)=0$~\cite{Wolf-2004-Zi}, which completes the proof.
\end{proof}
\noindent
Note the optimal simulator here is not produced such that when $\mbf{B}$ chooses to measure honestly they obtain that same $y$ received by the simulator (unless $|\mc{X}_{y}|=1$). Instead it just needs to produce a state that would yield $y'\in\mc{Y}_{x}\coloneqq\{y\in\mc{Y}\colon p(x,y)>0\}$, where $x$ is any element of $\mc{X}_{y}$. The constraint that $H(\X\searrow\Y|\Y)=H(\Y\searrow\X|\X)=0$ ensures that any of these samples are feasible given $p(x,y)$.

\subsection{\label{sec:BND-ME}QSNIR Privacy Upper Bounds Distinguishing Advantage}
More generally, we can lower bound $\delta$, via an operational distinguishing advantage. Let $P_{\guess}(\X|\Y)_{\textup{ideal}}$
denote the classical average guessing probability for $\X$, given $\Y$, derived from the classical minimum-entropy, $H_{\min}(\X|\Y)$. Likewise, let $P_{\textup{guess}}(\X|\mbf{B})_{\textup{real}}$ be the optimal MED success probability, given the ensemble formed from the reduced state of a coherent embedding with one half measured in the computational basis.
\begin{lemma}\label{lem:min-err-LB}
    For a perfectly correct ($\epsilon=0$) reduction from  $(\X,\Y)$ to the embedding $\op{\mrm{\Psi}_{\X\Y}}{\mrm{\Psi}_{\X\Y}}^{\mbf{AB}}$, the optimal $\delta$-privacy error is lower bounded by the following continuity-type inequality:
    \begin{equation}
       \delta_{\mbf{B}}\geq \delta_{\textup{MED},\mbf{B}}\coloneqq P_{\guess}(\X|\mbf{B})_{\textup{real}}-P_{\guess}(\X|\Y)_{\textup{ideal}}\geq 0.
    \end{equation}
\end{lemma}
\begin{proof}
    Let $\mc{S}$ be a fixed simulator. The corresponding privacy error is
    \begin{equation}
        \delta_{\mc{S},\mbf{B}}\coloneqq \td\left(\rho^{\X\mbf{B}}_{\textup{real}},\rho^{\X\mbf{B}}_{\textup{ideal}}\right)=\max_{0\preceq M\preceq \mbb{I}}\tr\left[M^{\X\mbf{B}}\left(\rho^{\X\mbf{B}}_{\textup{real}}-\rho^{\X\mbf{B}}_{\textup{ideal}}\right)\right],\label{eq:deltaS-as-max}
    \end{equation}
    with $\delta_{\mbf{B}}=\min_{\mc{S}}\delta_{\mc{S},\mbf{B}}$.
    Any feasible choice of $M^{\X\mbf{B}}$ lower bounds $\delta_{\mc{S},\mbf{B}}$. Consider the Block diagonal operator $M^{\X\mbf{B}}=\sum_{x}\op{x}{x}^{\X}\otimes M_{x}^{\mbf{B}}$, such that
    \begin{equation}
        \delta_{\mc{S},\mbf{B}}\geq\sum_{x}p(x)\left(\ip{\psi_{x}}{M_{x}|\psi_{x}}-\tr\left[M_{x}\mc{S}\left(\sum_{y}p(y|x)\op{y}{y}\right)\right]\right).
    \end{equation}
    If we further restrict $\{M_{x}\}_{x}$ to form a valid POVM, this solution remains feasible for the maximization in Eq.~\ref{eq:deltaS-as-max}, as $\sum_{x}M_{x}=\mbb{I}$ ensures that each $M_{x}\preceq \mbb{I}$. A valid choice is to fix this POVM such that it maximizes the first term in the sum above, and define 
    \begin{equation}
        P_{\textup{guess}}(\X|\mbf{B})_{\textup{real}}\coloneqq\sum_{x}p(x)\ip{\psi_{x}}{M_{x}|\psi_{x}},
    \end{equation}
    as the real world MED optimal guessing probability. While $\{M_{x}\}_{x}$ may be optimal to discriminate the ensemble $\{\{p(x),\ket{\psi_{x}}\}\}_{x}$, this need not be the case in the ideal world.
    In a mild abuse of notation, let $\mc{S}(\Y)$ denote the label for the state space output by the simulator and $ P_{\textup{guess}}(\X|\mc{S}(\Y))_{\textup{ideal}}$ the optimal MED guessing probability in the ideal world. We have that
    \begin{equation}
         P_{\textup{guess}}(\X|\mc{S}(\Y))_{\textup{ideal}}\geq \sum_{x}p(x)\tr\left[M_{x}\mc{S}\left(\sum_{y}p(y|x)\op{y}{y}\right)\right].
    \end{equation}
    Furthermore, by data processing, $P_{\textup{guess}}(\X|\Y)_{\textup{ideal}}\geq P_{\textup{guess}}(\X|\mc{S}(\Y))_{\textup{ideal}}$, such that $\delta_{\mc{S},\mbf{B}}\geq\delta_{\textup{MED},\mbf{B}}$. As this holds for any simulator, it must hold for the optimum, giving the lower bound $\delta_{\mbf{B}}\geq\delta_{\textup{MED},\mbf{B}}$. Finally, from Eq.~\ref{eq:dephasing-rw-ens}, $P_{\guess}(\X|\Y)_{\textup{ideal}}=P_{\guess}(\X|\mc{D}(\mbf{B}))_{\textup{real}}$ (in another abuse of notation). Therefore, again by data processing, $P_{\guess}(\X|\mbf{B})_{\textup{real}}\geq P_{\guess}(\X|\Y)_{\textup{ideal}}$ and $\delta_{\textup{MED},\mbf{B}}\geq0$.
\end{proof}
\noindent
This lemma demonstrates that $\delta_{\mbf{B}}$ upper bounds the following distinguishing task: \textit{With what advantage can $\mbf{B}$ determine, with minimum error, $\mbf{A}$'s sample of $\X$ from their reduced state of the embedding, given $\mbf{A}$'s honest measurement, over a classical sampling of $\Y$?} Conveniently, when an explicit simulator exists satisfying $\delta_{\mbf{B}}$, then $\mbf{B}$'s reduced state in the embedding provides no distinguishing advantage beyond this quantity, and an attack based on MED does not exist.  

In contrast, without defining an explicit simulator, Lem.~\ref{lem:min-err-LB} enables us to begin to analyze a lower bound on $\delta_{\mbf{B}}$ for specific target correlations, where we already know that perfect privacy is impossible. Let $\{M_{x}\}_{x}$ denote a POVM on $\mc{H}^{\mbf{B}}$ and $\mrm{\Lambda}$ be Hermitian, such that $\mrm{\Lambda}=\mrm{\Lambda}^{\dagger}$.
We can compute $\delta_{\textup{MED},\mbf{B}}$ explicitly, via the following SDP pair:
\begin{equation}\label{eq:sdp-med}
    \begin{aligned}
        &\underline{\textbf{Primal}} \\
        \textup{max} \quad & \sum_{x}p(x)\ip{\psi_{x}}{M_{x}|\psi_{x}} -\sum_{y}\max_{x}p(x,y), \\
        \textup{s.t.} \quad & \sum_{x}M_{x}=\mbb{I}, \\
        & M_{x}\succeq0,  & \quad \forall x\in\mc{X}.
        \\
        &\underline{\textbf{Dual}} \\
        \textup{min} \quad & \tr[\mrm{\Lambda}]-\sum_{y}\max_{x}p(x,y), \\
        \textup{s.t.} \quad & \mrm{\Lambda}\succeq p(x)\op{\psi_{x}}{\psi_{x}}, & \quad \forall x\in\mc{X} \\
        \quad & \mrm{\Lambda}=\mrm{\Lambda}^{\dagger}. 
    \end{aligned}
\end{equation}
This is the usual SDP formulation for MED, offset by a constant factor equal to the ideal optimal guessing probability.

It would be nice if the necessary and sufficient conditions on $p(x,y)$, such that $\delta_{\textup{MED}}=0$, aligned with the one-shot entropic quantities that satisfy $\delta_{\textup{QSNIR}}=0$. While $H(\X\searrow\Y|\Y)=0$ is a sufficient condition to fix $\delta_{\textup{MED}}=0$, we demonstrate in the next section that there exist gaps between these two privacy errors for certain correlations. In particular, for the family of correlations, $\msf{BSC}_{q}$, we demonstrate that $\delta_{\textup{MED}}(q)=0$, even when $H(\X\searrow\Y|\Y)>0$. We leave open the problem of constraining when $\delta_{\textup{MED}}=0$ solely in terms of $p(x,y)$. For the explicit optimal MED privacy errors that we compute in Sec.~\ref{sec:2PC}, we instead utilize lemmas discussed in App~\ref{app:med}, where applicable. 

\subsection{\label{sec:BND-SDP}Computing Nonzero Optimal QSNIR Privacy}
While MED gives a lower bound for QSNIR privacy, $\delta(\eqqcolon\delta_{\textup{QSNIR}})$ can be formulated as a convex optimization problem that can be determined exactly. The SDP pair that computes this privacy error is given as follows:
\begin{equation}\label{eq:sdp-qsnir}
    \begin{aligned}
        &\underline{\textbf{Primal}} \\
        \textup{min} \quad & \frac{1}{2}\sum_{x}\tr[Q_{x}^{+}+Q_{x}^{-}], \\
        \textup{s.t.} \quad & Q_{x}^{+}-Q_{x}^{-}=p(x)\op{\psi_{x}}{\psi_{x}}-\sum_{y}p(x,y)J_{y}, & Q_{x}^{+},Q_{x}^{-}\succeq0, & \quad \forall x\in\mc{X}, \\
        & \tr[J_{y}]=1,
        & J_{y}\succeq0, & \quad \forall y\in\mc{Y}. \\
        &\underline{\textbf{Dual}} \\
        \textup{max} \quad &  \sum_{x}p(x)\ip{\psi_{x}}{W_{x}|\psi_{x}}-\sum_{y}\lambda_{y},  \\
        \textup{s.t.} \quad & 0 \preceq W_{x}\preceq \mbb{I}, & & \quad \forall x\in\mc{X}, \\
        & \lambda_{y}\mbb{I}\succeq\sum_{x}p(x,y)W_{x}, & \lambda_{y}\in\mbb{R}, & \quad \forall y\in\mc{Y}.
    \end{aligned}
\end{equation}
Below we detail how to interpret the primal and dual problem in the context of QSNIR, and prove that strong duality holds.

The goal of the primal problem in the SDP is to determine the simulator that optimally recovers from the dephasing map applied by the honest party.
We define the simulator as a preparation channel, $\mc{S}\colon\D(\mbb{C}^{|\mc{Y}|})\to\D(\mc{H}^{\mbf{B}})$, fixed by the Choi matrix~\cite{Choi-1975-Cp}
\begin{equation}
    J_{\mc{S}}^{\mbf{B}\Y}=\sum_{y}J_{y}^{\mbf{B}}\otimes\op{y}{y}^{\Y},
\end{equation}
such that
\begin{equation}
    \mc{S}^{\mbf{B}}\left(\op{y}{y}^{\Y}\right)=\tr_{\Y}\left[\left(\mbb{I}^{\mbf{B}}\otimes\op{y}{y}^{\Y}\right)J_{\mc{S}}^{\mbf{B}\Y}\right]=J_{y}^{\mbf{B}}, \quad \forall y.
\end{equation}
In words, the simulator receives a classical sample $y\samp\Y$ and outputs the state $J_{y}$ on the corrupted party's registers. In general, computing $\min_{\sigma}\td(\rho,\sigma)$, where $\sigma$ is constrained to be in some set, requires finding the state that minimizes the trace of the positive (or, equivalently, the negative) part of the operator $\mrm{\Delta}\coloneqq\rho-\sigma$. When computing QSNIR privacy errors, we need to perform this minimization over a whole set of distances, one for each element in $\{\op{\psi_{x}}{\psi_{x}}\}_{x}$, minimized against averages over all the possible outputs from the simulator, given by $\{J_{y}\}_{y}$. This introduces another set of convex constraints into the SDP, simply to ensure that each $J_{y}$ is a valid quantum state. 

Conversely, in the dual SDP, any mention of the simulator drops out of the problem completely.\footnote{In App.~\ref{app:QSNIR-SDP-dual}, we formally demonstrate how the dual is constructed.} 
We can interpret each $W_{x}$ as a \textit{witness} of 
a particular distinguishing (or cheating) strategy, for the sample, $x\samp\X$, that the honest party obtains.  
How well this strategy scores is offset by what an adversary would already know about $x$ in the ideal world given $y$. This offset is minimized when each $\lambda_{y}=\lambda_{\max}(\sum_{x}p(x,y)W_{x})$, where $\lambda_{\max}(\cdot)$ is the maximum eigenvalue function. Where the true optimum of the dual differs from MED, is that in general $\sum_{x}W_{x}\neq\mbb{I}$, such that this ``cheating strategy'' does not describe a valid POVM for the corrupted party, as is the case with MED.
Notably, if we do enforce that the witnesses together form a POVM, which we redefine as $\{M_{x}\}_{x}$, then
\begin{equation}
    \lambda_{\max}\left(\sum_{x}p(x,y)M_{x}\right)\leq\max_{x}p(x,y),
\end{equation}
such that the true optimum is lower-bounded as 
\begin{align}
    \delta &\geq \max_{\{0\preceq W_{x}\preceq \mbb{I}\}_{x}} \sum_{x}p(x)\ip{\psi_{x}}{W_{x}|\psi_{x}}-\sum_{y} \lambda_{\max}\left(\sum_{x}p(x,y)W_{x}\right)\nonumber\\
    &\geq \max_{\{0\preceq M_{x}\}_{x},\sum_{x}M_{x}=\mbb{I}} \sum_{x}p(x)\ip{\psi_{x}}{M_{x}|\psi_{x}}-\sum_{y}\max_{x}p(x,y)  = \delta_{\textup{MED}},
\end{align}
recovering Lem.~\ref{lem:min-err-LB}. 
Instead, we group $\{\{W_{x},\mbb{I}-W_x\}\}_{x}$ for each $x$ to form a Helstrom measurement that optimizes the trace distance between $\op{\psi_{x}}{\psi_{x}}$ and $\sum_{y}p(y|x)J_{y}$, the quantum state held by the corrupt party in the real and ideal world, respectively, conditioned on $x\samp\X$. Rather than describing a measurement made by this corrupt party, QSNIR privacy quantifies the average over all $x$ of the binary discriminating power of the environment to distinguish between the real and ideal worlds. The reason why an explicit simulator drops out of the dual problem is essentially by data processing. Instead of asking how well we can recover from the dephasing of the honest party, such that cheating is not possible, the dual problem asks how well we can cheat given all possible recovery maps.

If the primal problem and dual problem have strong duality, then $\delta$ becomes an operationally meaningful quantifier of simulation-based (and therefore composable) privacy. The simulator in the primal problem works to prevent the environment from distinguishing whether a deviation occurred between the real and ideal worlds. The dual problem shows that the inability of the environment to make this distinction, bounds how well an adversary can cheat. The following lemma states this duality for any correlation.
\begin{lemma}[Strong Duality and Complementary Slackness]\label{lem:QSNIRsdp-duality}
    Strong duality holds for the SDP pair that computes $\delta_{\textup{QSNIR}}$. Given primal-feasible $\{Q_{x}^{\pm}\}_{x}$ and $\{J_{y}\}_{y}$, along with dual-feasible $\{W_{x}\}_{x}$ and $\{\lambda_{y}\}_{y}$, complementary slackness is expressed through the constraints
    \begin{subequations}
        \begin{align}
             Q_{x}^{+}(\mbb{I}-W_{x})=0, \quad Q_{x}^{-}W_{x}&=0, \quad \forall x, \\
             \left(\lambda_{y}\mbb{I}-\sum_{x}p(x,y)W_{x}\right)J_{y}&=0, \quad \forall y.
        \end{align}
    \end{subequations}
\end{lemma}
\noindent
See App.~\ref{app:lem-duality} for a proof via Slater's condition.

\section{\label{sec:2PC}One-Shot Optimal Privacy Errors for Universal 2PC Correlations}
We now employ the results of the previous section to compute the optimal privacy errors in Tab.~\ref{tab:privacy-errors}. 
In App.~\ref{app:comp-ent}, we detail how to compute $H(\X\wedge\Y)$ and $H(\X\searrow \Y|\Y)$ (along with its reverse quantity) for each correlation. These entropies serve as a guiding tool for understanding where perfect privacy is and isn't admitted. 
For $\msf{BSC}_{q}$ and $\msf{BEC}_{q}$ correlations, the values of $q$ that satisfy perfect privacy correspond to fixed points in the behavior of the respective channel that the correlation is derived from.  As we focus on binary correlations in this section, we denote addition and multiplication mod 2 by $\oplus$ and $\wedge$, respectively.

\subsection{\label{sec:2PC-SK}Binary Symmetric Key}
A binary symmetric key~\cite{Shannon-1949-Ct} is represented by a correlation $\msf{SK}=(\X,\Y)$, fixed by the parent pmf
\begin{equation}
    p_{\msf{SK}}(x,y)\coloneqq
    \begin{cases}
        1/2, & x\oplus y=0, \\
        0, & \textup{else}.
    \end{cases}
\end{equation}
With $\msf{SK}$, it is straightforward to see that
\begin{equation}
    I(\X\colon\Y)-H(\X\wedge\Y)=H(\X\searrow\Y|\Y)=H(\Y\searrow\X|\X)=0,
\end{equation}
as $\X=\Y=\X\wedge\Y$, explicitly.
\begin{corollary}[to Thm.~\ref{thm:perf-priv}]\label{cor:SK-priv}
    There exists a simple $0$-QSNIR from $\msf{SK}$ to its embedding, $\ket{\mrm{\Psi}_{\X\Y}}^{\mbf{AB}}=\ket{\mrm{\Phi}^{+}}^{\mbf{AB}}$, which is the maximally entangled state $\ket{\mrm{\Phi}^{+}}\coloneqq(\ket{00}+\ket{11})/\sqrt{2}$.
\end{corollary}
\noindent
In fact, the simulator that achieves optimality is the identity, as the reduced state of the corrupted party is the maximally mixed state, $\mbb{I}/2$, which is invariant under the dephasing channel they would perform if honest. 

\subsection{\label{sec:2PC-OK}Binary Oblivious Key}
Binary oblivious key, denoted by $\msf{OK}=((\msf{W}_{0},\msf{W}_{1}),(\msf{V},\msf{S}))$, is a four bit correlation which enables binary $\binom{2}{1}$-OT, via only two rounds of public communication~\cite{Wolf-2004-OT}.\footnote{One can similarly interpret $\msf{OK}$ as the output of $\binom{2}{1}$-OT with randomized inputs, where $(w_{0},w_{1})$ are $\mbf{A}$'s message bits, $s$ is $\mbf{B}$'s selection bit, and $v$ is the transferred bit that $\mbf{B}$ receives. In setting $w_{s}\coloneqq v$, it is evident that the correct bit is transferred.}  Explicitly, this correlation is fixed by the parent pmf
\begin{equation}\label{eq:ok-dist}
    p_{\msf{OK}}(w_0,w_1,v,s)\coloneqq
    \begin{cases}
        1/8, & (w_{0}\wedge(s\oplus1))\oplus(w_{1}\wedge s)=v, \\
        0, & \textup{else}.
    \end{cases}
\end{equation}
For $\msf{OK}$, we cannot have perfect privacy as, 
\begin{equation}
    H\left((\msf{W}_{0},\msf{W}_{1})\searrow(\msf{V},\msf{S})|(\msf{V},\msf{S})\right)=1,
\end{equation}
and likewise for the reverse case. This also has a nice interpretation. As the correlation represents the outputs of binary randomized $\binom{2}{1}$-OT, $\mbf{A}$ knows that $\mbf{B}$ possesses one of their two bits, but there is a bit of uniform randomness preventing them from guessing which one. Likewise, $\mbf{B}$ knows the value of one of $\mbf{A}$'s bits, but the other remains uniformly random. 

We demonstrate in the theorem below that the MED and QSNIR privacy errors for this correlation are the same for either a corrupt $\mbf{A}$ or $\mbf{B}$. This symmetry is easier to see when considering an alternative four bit correlation, a binary multiplication triple,\footnote{Triples are a universal resource that enables private multiplication in the pre-processing model~\cite{Beaver-1992-EM}.}
which is equivalent to $\msf{OK}$ by an invertible $0$-SNIR~\cite{Wolf-2005-Ot}.
\begin{theorem}[Optimal Privacy Errors for $\msf{OK}$ and $\msf{Trip}$]\label{thm:priv-OK}
    The optimal MED privacy error for $\msf{Trip}$ (and consequently $\msf{OK}$) is
    \begin{equation}
        \left(\delta^{\msf{OK}}_{\textup{MED}} = \right)\delta^{\msf{Trip}}_{\textup{MED}} = \frac{1}{8}\left(2\sqrt{2}-1\right) \approx 0.229.
    \end{equation}
    In contrast, the optimal QSNIR privacy error that guarantees composability is
    \begin{equation} 
        \left(\delta_{\textup{QSNIR}}^{\msf{OK}} = \right)\delta_{\textup{QSNIR}}^{\msf{Trip}} = \frac{1}{4}\left(\sqrt{5}-1\right) \approx 0.309.\footnotemark
    \end{equation}
    By symmetry, these errors are the same for either a corrupt $\mbf{A}$ or $\mbf{B}$.
\end{theorem}
\footnotetext{While further context is found within the proof, the golden ratio has reared its head. The polynomial $\varphi^{2}-\varphi-1=0$ has roots $\varphi_{+}=(1+\sqrt{5})/2$, the golden ratio, and $\varphi_{-}=(1-\sqrt{5})/2$, which are related by $\varphi_{+}=-1/\varphi_{-}$. In our context, the optimal QSNIR privacy error that we compute for $\msf{OK}$ is proportional to the negative root.}
\noindent
See App.~\ref{app:thm-priv-OK} for a detailed proof.

Understanding the root of these privacy errors relates directly to a cyclic invariance encoded in the underlying correlation. In Fig.~\ref{fig:OK-TRIP}, we detail the corresponding connecting graphs for either correlation, depicted both in a bipartite and cyclic format. This forces the Gram matrix of the adversary's real world ensemble to be circulant, which lets us employ the pretty-good measurement (PGM) in order to compute $\delta^{\msf{Trip}}_{\textup{MED}}$. The same symmetry informs our proof of $\delta^{\msf{Trip}}_{\textup{QSNIR}}$, where we construct families of primal and dual feasible solutions whose optima align.

\begin{figure}[t]
    \centering
    \includegraphics[width=0.95\linewidth]{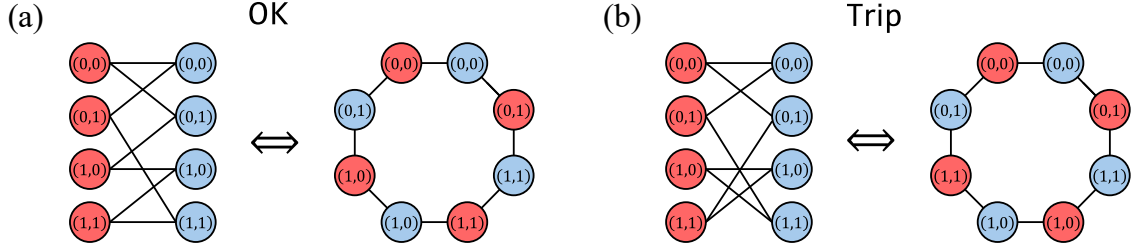}
    \caption{Connecting Cycle Graphs for (a) $\msf{OK}$ and (b) $\msf{Trip}$. Red vertices denote samples belonging to $\mbf{A}$, while blue vertices belong to $\mbf{B}$.}
    \label{fig:OK-TRIP}
\end{figure}

\subsection{\label{sec:2PC-BSC}Binary Symmetric Correlations} This family of correlations describes the joint distribution over the input and output of a binary symmetric channel, where $x$ is the input and the output is $y=x$ with probability $1-q$.
We denote $\msf{BSC}_{q}\coloneqq(\X,\Y)$ as a two-bit correlation, given by the parent pmf
\begin{equation}
    p_{\msf{BSC}_{q}}(x,y)\coloneqq
    \begin{cases}
        (1-q)/2, & x\oplus y=0, \\
        q/2, & \textup{else}.
    \end{cases}
\end{equation}
Note that $\msf{BSC}_{0}=\msf{SK}$ (and $\msf{BSC}_{1}$ is equal to antisymmetric key). Similarly, $\msf{BSC}_{1/2}$ is just a uniform product distribution. Therefore, $\delta_{\textup{QSNIR}}^{\msf{BSC}}(q)=0$ only when $q\in\{0,1/2,1\}$, in line with when $H(\X\searrow\Y|\Y)=0$.

\begin{theorem}[Optimal Privacy Errors for $\msf{BSC}_{q}$]\label{thm:BSC-priv}
    The family of correlations, $\msf{BSC}_{q}$, is perfectly private against MED-based attacks ($\delta_{\textup{MED}}=0$), but the optimal QSNIR privacy error that guarantees composability can be reduced to the following scalar optimization problem:
    \begin{equation}\label{eq:bsc-qsnir-priv-err}
          \delta^{\msf{BSC}}_{\textup{QSNIR}}(q) = \frac{1}{2}\min_{0\le\theta\le\pi/2}\bigg[\sqrt{
        (s-\sin\theta)^{2}+t^{2}(1-\cos\theta)^{2}}\bigg],
    \end{equation}
    where $t=1-2\bar{q}$ and $s=2\sqrt{\bar{q}(1-\bar{q})}$, with $\bar{q}=\min\{q,1-q\}\in[0,1/2]$.
    These errors are the same for either a corrupt $\mbf{A}$ or $\mbf{B}$. 
\end{theorem}
\noindent
In App.~\ref{app:thm-BSC-priv}, we detail the proof for each of these privacy errors.  See Fig.~\ref{fig:BSC-BEC} for a plot.

Interestingly, this theorem tells us that while MED-based attacks reveal no information beyond the ideal setting, guaranteeing simulation-based privacy necessitates allowing for nonzero error, for all nontrivial $q$. Eq.~\ref{eq:bsc-qsnir-priv-err} defines the problem of finding the shortest distance in $\mbb{R}^{2}$ between the point $(s,t)$, which lies on the boundary of the $XZ$-plane of the Bloch sphere, and an ellipse formed by $\sin^{2}\theta+t^{-2}\cos^{2}\theta=1$, inside the Bloch sphere. The simulator prepares pure states with Bloch vectors, $(\sin\theta,\pm\cos\theta)$, but when averaged in the ideal world, this vector gets contracted along the $Z$-axis by a factor of $t$, related to how well the classical bit-flip error can be corrected in the channel. 
The pair of vectors on the boundary of this ellipse that are closest to $(s,\pm t)$ fixes a particular $\theta$ in the solution, giving the angle between the states that the simulator needs to prepare.

\begin{figure}[t]
    \centering
    \includegraphics[width=0.9\linewidth]{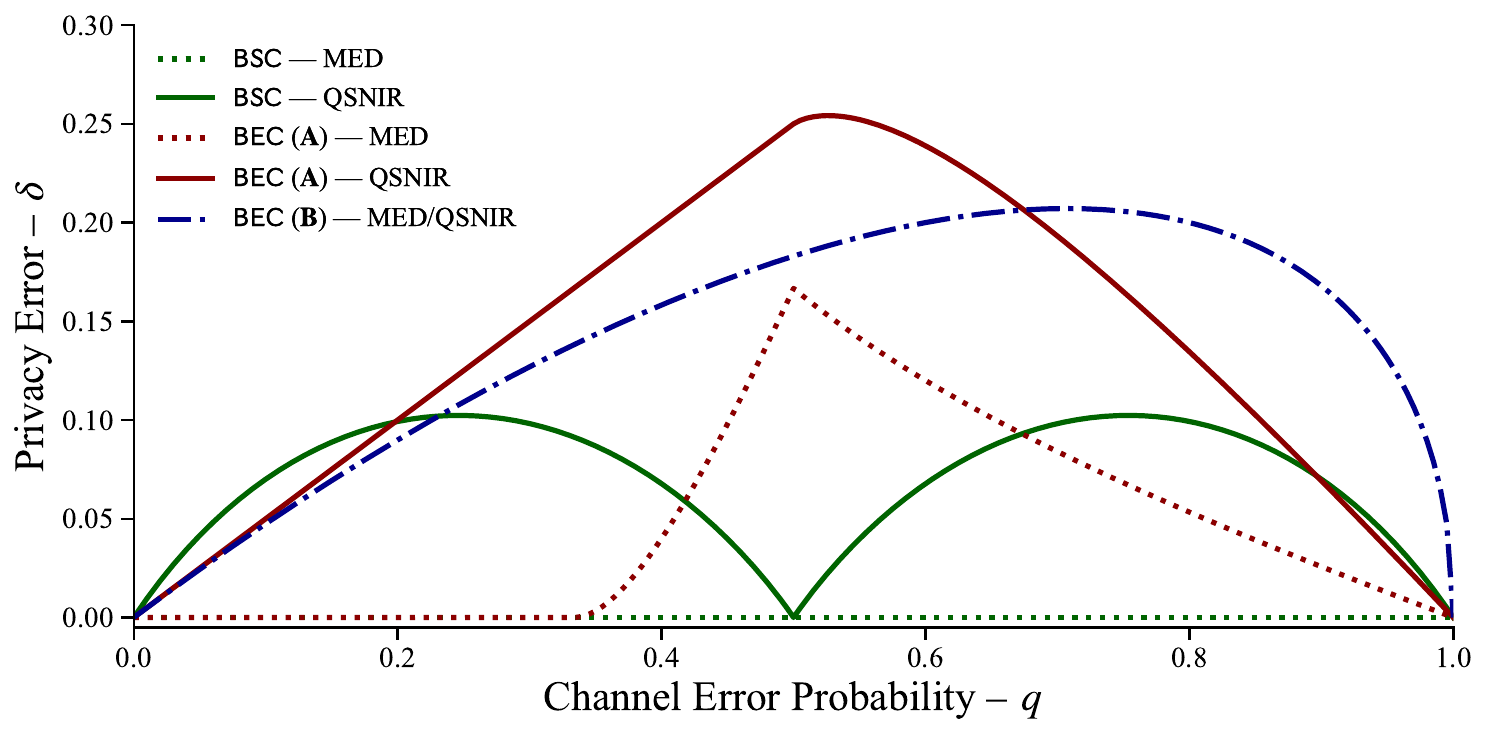}
    \caption{Optimal MED (dotted lines) and QSNIR (solid lines) privacy errors for the families of correlations $\msf{BSC}_{q}$ and $\msf{BEC}_{q}$, over $q\in[0,1]$. Both privacy errors for $\msf{BSC}_{q}$ are equivalent between $\mbf{A}$ and $\mbf{B}$ (green), while distinct errors exist for $\msf{BEC}_{q}$ when either $\mbf{A}$ is corrupt (red), or $\mbf{B}$ is corrupt (blue). (As the MED and QSNIR errors are equal for $\msf{BEC}_{q}$ when $\mbf{B}$ is corrupt, we denote it by a dash-dot line.)}
    \label{fig:BSC-BEC}
\end{figure}

\subsection{\label{sec:2PC-BEC}Binary Erasure Correlations}
The family of correlations $\msf{BEC}_{q}\coloneqq(\X,\Y)$ describes the joint distribution over the inputs and outputs of a binary erasure channel. Similar to the family of symmetric correlations, $x=y$ with probability $(1-q)/2$, however the rest of the weight in the distribution corresponds to $y=\perp$, independent of $x$, describing the case when $x$ is input into the erasure channel but lost in the output (and, hence, $|\mc{Y}|=3$ here). The pmf that describes this family is
\begin{equation}
    p_{\msf{BEC}_{q}}(x,y)\coloneqq
    \begin{cases}
        (1-q)/2, & x\oplus y=0\\
        q/2, &  y=\perp \\
        0, & \textup{else}.
    \end{cases}
\end{equation}
Note again that $\msf{BEC}_{0}=\msf{SK}$. Perfect QSNIR privacy can only hold when $q\in\{0,1\}$, i.e. where there is no error or where the bit is always lost in the channel, such that there is no correlation to begin with.

Unlike the previous correlations we discussed, the MED and QSNIR privacy errors that we compute are distinct, when either $\mbf{A}$ or $\mbf{B}$ is corrupt.
\begin{theorem}[Optimal Privacy Errors for $\msf{BEC}_{q}$]\label{thm:BEC-priv}
    For the family of correlations, $\msf{BEC}_{q}$, when $\mbf{A}$ is corrupt, we have distinct MED and QSNIR privacy errors for all $q\in(0,1)$. In the context of MED-based attacks, the optimal distinguishing advantage that $\mbf{A}$ can achieve is
    \begin{equation}
        \delta^{\msf{BEC}}_{\textup{MED},\mbf{A}}(q) \coloneqq \begin{cases}
            0, &  0 \leq q \leq 1/3,\\
            g(q), & 1/3 < q \leq 1,
        \end{cases}     
    \end{equation}
    where
    \begin{equation}
        g(q)\coloneqq 4q^2/(5q-1)-\max\{q,1-q\}.
    \end{equation}
    For QSNIR privacy, we find a distinct piecewise solution
    \begin{equation}
        \delta^{\msf{BEC}}_{\textup{QSNIR},\mbf{A}}(q) \coloneqq
        \begin{cases}
            q/2, & 0\leq q\leq1/2,\\
            f(q), & 1/2< q\leq1,
        \end{cases}
    \end{equation}
    where
    \begin{equation}
        f(q)\coloneqq (1-q)\sin\frac{\theta_{0}}{2}+\frac{q}{2}(1-\sin\theta_{0}), \quad\theta_{0}\coloneqq4\tan^{-1}\sqrt{3q-\sqrt{1+q(9q-2)}}.
    \end{equation}

   Instead, when $\mbf{B}$ is corrupt, we have that $\delta^{\msf{BEC}}_{\textup{MED},\mbf{B}}(q)=\delta^{\msf{BEC}}_{\textup{QSNIR},\mbf{B}}(q) \eqqcolon \delta^{\msf{BEC}}_{\mbf{B}}(q)$, with
    \begin{equation}
        \delta^{\msf{BEC}}_{\mbf{B}}(q) \coloneqq \frac{1}{2}\left(\sqrt{1-q^{2}}-(1-q)\right),
    \end{equation}
    such that the gap between MED and QSNIR privacy is tight for all $q\in[0,1]$.
\end{theorem}
\noindent
A proof is given in App.~\ref{app:them-BEC-priv}, and Fig.~\ref{fig:BSC-BEC} depicts a plot for each error.  Below, we comment on these solutions to motivate how this family, and other correlation families like it, may be used to further understand the gap between MED and QSNIR privacy. 

The simpler case is when $\mbf{B}$ is corrupt, and the MED and QSNIR privacy errors align. This follows from how projectors which form for the PVM that maximizes the MED guessing probability are among those that the simulator prepares for optimal QSNIR privacy. We might expect this solution to be concave, with endpoints where the privacy errors vanish. It is interesting, however, that the maximum skews past $q=1/2$, to $q=1/\sqrt{2}$. Operationally, this says $\mbf{B}$ can optimally cheat when it's more likely than not that the bit $y$ is uncorrelated from $x$ (i.e. $x$ is lost in the erasure channel). While this is indeed where the most information is leaked by the embedding, it may not be the place that the correlation is the most operationally useful for 2PC.

Instead, when $\mbf{A}$ is corrupt, each privacy error is given by a distinct piecewise solution, where even the intervals of each piece don't line up between the errors. Here, the MED privacy is also no longer concave over the interval for $q$, but is maximized at exactly $q=1/2$. The QSNIR privacy error is maximized at a point just past $q=1/2$ (at $q\approx0.526$). Where a cheating $\mbf{A}$ might be able to discriminate their real world ensembles, exactly where it is just as likely as not that $x$ is lost, the environment is able to distinguish further if it is slightly more likely than not that the bit is lost. 

\section{\label{sec:DSC}Discussion}
In this work, we introduced a simulation-based framework that handles transforming quantum states into classical cryptographic correlations from only local measurements. QSNIRs extend classical SNIRs to the quantum setting while preserving statistical security guarantees against a dishonest party with quantum side-information.
Within this framework, we showed the limitations of entanglement as a non-interactive cryptographic resource, with an exactly computable one-shot error in privacy. Several directions are opened by this work, including further quantifying the gap between MED and QSNIR privacy, and formally minimizing these privacy errors over encoded phase information. 

With respect to composability, we proved that QSNIRs compose in parallel and sequence with each other. We believe the simulation-based structure of this framework can then be made amenable to composition within large composable quantum cryptographic frameworks~\cite{Portmann-2022-SQ,Broadbent-2023-Cc}, such as the \textit{abstract cryptography} framework~\cite{Maurer-2012-CC}. This motivates studying the construction of quantum protocols that distribute correlations ``from scratch'', by leveraging \textit{self-testing} in conjunction with a QSNIR. 
Device-independent self-testing describes the process of certifying a quantum resource from statistics alone, making no assumptions about the devices used~\cite{Coladangelo-2017-Ap,Bowles-2018-DE,Baccari-2020-SB,Murta-2023-SD,Balanzó-Juandó-2026-Ap}, and serves as the basis for device-independent QKD (DIQKD)~\cite{Ekert-2014-up,Vazirani-2014-Fd}.  
While in the 2PC setting, statistically vanishing privacy errors are impossible for universal correlations, our QSNIR framework can be adapted to an honest majority multipartite setting (where self-testing has been shown to remain possible~\cite{Baccari-2020-SB,Balanzó-Juandó-2026-Ap}) in order to construct feasible reductions from these universal correlations, e.g. using multipartite resources such as graph states~\cite{Gold-2026-Hp}. 

Finally, we also comment that the QSNIR privacy problem admits a statistical interpretation through Blackwell's comparison of experiments~\cite{Blackwell1951,Blackwell1953,LeCam96}. For either corruption, the honest party's output indexes a statistical experiment. In this context, the ideal functionality gives the adversary a classical experiment, whereas the coherent implementation gives a quantum experiment. Perfect privacy asks whether the latter is a post-processing (garbling) of the former, while the optimal QSNIR privacy error measures an approximate quantum deficiency between the two. The MED privacy error corresponds to one particular decision problem in this comparison, explaining both the universal lower bound in guessing advantage and why, for some correlations, the bound is not tight. Further understanding QSNIR privacy through this lens may enable us to answer open questions posed by this work.

\subsubsection*{Acknowledgments} This work was supported by the NSF Quantum Leap Challenge Institute on Hybrid Quantum Architectures and Networks (NSF Award No. 2016136). 

\clearpage

\printbibliography

\clearpage

\appendix
\section{\label{app:med}Computing MED Guessing Probabilities}
Let $\mbf{B}$ be the label for the discerning party.
In general the success probability of MED for a particular POVM, $\{M_{x}\}_{x}$, is given as
\begin{equation}
    P_{\textup{succ}}(\X|\mbf{B})_{\rho} \coloneqq \sum_{x}p(x)\tr\left[M_{x}\rho_{x}\right],
\end{equation}
which we aim to maximize. Formally, this maximization is given by an SDP pair with strong duality:
\begin{equation}\label{eq:sdp-med-gen}
    \begin{aligned}
        &\underline{\textbf{Primal}} \\
        \textup{max} \quad & \sum_{x}p(x)\ip{\psi_{x}}{M_{x}|\psi_{x}}, \\
        \textup{s.t.} \quad & \sum_{x}M_{x}=\mbb{I}, \\
        & M_{x}\succeq0,  & \quad \forall x\in\mc{X}.
        \\
        &\underline{\textbf{Dual}} \\
        \textup{min} \quad & \tr[\mrm{\Lambda}], \\
        \textup{s.t.} \quad & \mrm{\Lambda}\succeq p(x)\op{\psi_{x}}{\psi_{x}}, & \quad \forall x\in\mc{X} \\
        \quad & \mrm{\Lambda}=\mrm{\Lambda}^{\dagger}. 
    \end{aligned}
\end{equation}
For a POVM that achieves the optimum, we relabel $P_{\textup{guess}}\coloneqq P_{\textup{succ}}$ as the optimal probability to guess $x\samp\X$ given $\mbf{B}$'s quantum side information. In~\cite{Konig-2009-OM}, it was shown that this optimum is related to $H_{\min}(\X|\mbf{B})_{\rho}$, as
\begin{equation}
    P_{\textup{guess}}(\X|\mbf{B})_{\rho}=2^{-H_{\min}(\X|\mbf{B})}.
\end{equation}
This yields an interpretation to the quantum min-entropy for cq states, as a measure of the optimal success probability for MED.

Complementary slackness between the SDP pair in Eq.~\ref{eq:sdp-med-gen} allows for the existence of necessary and sufficient conditions for a POVM to be optimal. These are the well-known \textit{Kennedy-Yuen-Lax conditions}, which we detail in the following lemma.
\begin{lemma}[Kennedy-Yuen-Lax~\cite{Yuen-1975-Ot}]\label{lem:kyl}
    The POVM $\{M_{x}\}_{x}$ maximizes the MED success probability for the ensemble $\{\{p(x),\rho_{x}\}\}_{x}$, if and only if
    \begin{equation}
        \sum_{x'}p(x')M_{x'}\rho_{x'}-p(x)\rho_{x}\succeq0, \quad \forall x.
    \end{equation}
\end{lemma}
\noindent
Given a candidate POVM, this lemma gives a systematic way to check for its optimality.

Some of the earliest work on MED regards ensembles of limited size and yields analytical results. In the case of a two-state, or binary, ensemble, let $p_{x}\coloneqq p(x)$ for $x\in\mbb{Z}_{2}$.
The optimal success probability can be computed as follows.
\begin{lemma}[Holevo-Helstrom~\cite{Holevo-1973-Sd,Helstrom-1969-Qd}]\label{lem:helevo}
    Given a binary ensemble consisting of states $\{\rho_{0},\rho_{1}\}$, with priors $\{p_{0},p_{1}\}$, the optimal MED guessing probability is
    \begin{equation}
        P_{\textup{guess}}=\frac{1}{2}+\td\left(p_{0}\rho_{0},p_{1}\rho_{1}\right).
    \end{equation}
\end{lemma}
\noindent
The binary POVM, $\{M,\mbb{I}-M\}$, that attains this optimum is known as a \textit{Helstrom measurement}, found by projecting onto the positive and negative eigenspaces of $\mrm{\Delta}\coloneqq p_{0}\rho_{0}-p_{1}\rho_{1}$.

Lastly, the Pretty-Good Measurement (PGM) (or Square-Root Measurement)~\cite{Hausladen-1994-PG,Eldar-2001-qd}, denotes a particular fixed strategy for MED, which can exploit symmetries baked into the ensemble. Formally, the PGM involves defining a POVM, $\{M_{x}\}_{x}$, with elements
\begin{equation}
    M_{x}\coloneqq p(x)(\sqrt{\rho})^{-1}\rho_{x}(\sqrt{\rho})^{-1},
\end{equation}
where $\rho\coloneqq\sum_{x}p(x)\rho_{x}$ is the average state of the ensemble, and $(\sigma)^{-1}$ denotes the \textit{Moore-Penrose pseudoinverse} of $\sigma$. The following lemma demonstrates a sufficient condition for the optimality of this POVM for MED on pure states.
\begin{lemma}[Optimality of the PGM~\cite{Eldar-2001-qd}]\label{lem:PGM}
Let $\{\{p(x),\ket{\psi_{x}}\}\}_{x}$ be a pure state ensemble, where $p(x)=1/|\mc{X}|$ for all $x$ (uniform priors), and $\ket{\psi_{x}}=U^{x}\ket{\psi_{0}}$ for some unitary $U$ and fiducial state $\ket{\psi_{0}}$ within the ensemble (geometrically uniform states). Then, the PGM is optimal for MED.
\end{lemma}
\noindent
To determine the optimum, we define the ensemble's corresponding \textit{Gram matrix}, $G$, with matrix elements $G^{(x,x')}=\ip{\psi_{x}}{\psi_{x'}}$, and compute
\begin{equation}
    P_{\textup{guess}}=\frac{1}{|\mc{X}|^{2}}\left(\sum_{x}\sqrt{\lambda_{G}^{(x)}}\right)^{2},
\end{equation}
where $\lambda_{G}$ is the set of eigenvalues for $G$, and $\lambda_{G}^{(x)}$ is the $x^{\textup{th}}$ entry. 

\section{\label{app:thm-para}Proof of Thm.~\ref{thm:parallel} (Parallel Composition)}
The proof for correctness and privacy of the composite reduction follow identically. For simplicity we only detail with respect to privacy against $\mbf{B}$. 
Individually, we have that
\begin{subequations}
    \begin{equation}
         \td\left(\left(\mc{N}^{\mbf{A}}_{1}\otimes \mrm{id}^{\mbf{B}}\right)\left(\rho^{\mbf{AB}}_{1,\src}\right),\left(\mrm{id}^{\mbf{A}}\otimes\mc{S}^{\mbf{B}}_{1}\right)\left(\rho^{\mbf{AB}}_{1,\tgt}\right)\right)\leq\delta_{\mbf{B},1},
    \end{equation}
    \begin{equation}
        \td\left(\left(\mc{N}^{\mbf{A}}_{2}\otimes \mrm{id}^{\mbf{B}}\right)\left(\rho^{\mbf{AB}}_{2,\src}\right),\left(\mrm{id}^{\mbf{A}}\otimes\mc{S}^{\mbf{B}}_{2}\right)\left(\rho^{\mbf{AB}}_{2,\tgt}\right)\right)\leq\delta_{\mbf{B},2},
    \end{equation}
\end{subequations}
where $\mc{S}^{\mbf{B}}_{1}$ and $\mc{S}^{\mbf{B}}_{2}$ are corresponding simulators for each reduction. Therefore, if we define a simulator $\mc{S}^{\mbf{B}}=\mc{S}^{\mbf{B}}_{1}\otimes\mc{S}^{\mbf{B}}_{2}$, and denote $\mc{N}^{\mbf{A}}=\mc{N}_{1}^{\mbf{A}}\otimes\mc{N}_{2}^{\mbf{A}}$, it must be that 
\begin{equation}\label{eq:para-hybrid-contr}
    \td\left(\left(\mc{N}^{\mbf{A}}\otimes \mrm{id}^{\mbf{B}}\right)\left(\rho^{\mbf{AB}}_{1,\src}\otimes\rho^{\mbf{AB}}_{2,\src}\right), \left(\mrm{id}^{\mbf{A}}\otimes\mc{S}^{\mbf{B}}\right)\left(\rho^{\mbf{AB}}_{1,\tgt}\otimes\rho^{\mbf{AB}}_{2,\tgt}\right)\right)\leq\delta_{\mbf{B},1}+\delta_{\mbf{B},2}.
 \end{equation}
 by the sub-additivity of the trace distance under tensor product, when applied to quantum states. \qed

\section{\label{app:thm-seq}Proof of Thm.~\ref{thm:sequential} (Sequential Composition)}
We employ a similar argument as the previous proof, again focusing on only privacy against $\mbf{B}$. 
We set $\mc{S}^{\mbf{B}}=\mc{S}^{\mbf{B}}_{1}\circ\mc{S}^{\mbf{B}}_{2}$ as the sequential compositions of $\mc{S}^{\mbf{B}}_{1}$ and $\mc{S}^{\mbf{B}}_{2}$ (the corresponding simulators satisfying $\delta_{\mbf{B},1}$ and $\delta_{\mbf{B},2}$ privacy of the individual reductions), and $\mc{N}^{\mbf{A}}=\mc{N}_{2}^{\mbf{A}}\circ\mc{N}_{1}^{\mbf{A}}$. We then extend the composition utilizing this simulator,
\begin{align}
    &\td\left(\left(\mc{N}^{\mbf{A}}\otimes \id^{\mbf{B}}\right)\left(\rho^{\mbf{AB}}_{0,\src}\right), \mrm{id}^{\mbf{A}}\otimes\mc{S}^{\mbf{B}}\left(\rho^{\mbf{AB}}_{2,\tgt}\right)\right)\nonumber\\
    &\quad\quad\quad\quad\quad\quad=\td\left(\left(\mc{N}_{2}^{\mbf{A}}\circ\mc{N}_{1}^{\mbf{A}}\otimes \mrm{id}^{\mbf{B}}\right)\left(\rho^{\mbf{AB}}_{0,\src}\right), \left(\mrm{id}^{\mbf{A}}\otimes\mc{S}^{\mbf{B}}_{1}\circ\mc{S}^{\mbf{B}}_{2}\right)\left(\rho^{\mbf{AB}}_{2,\tgt}\right)\right)\nonumber\\
    &\quad\quad\quad\quad\quad\quad\leq\td\left(\left(\mc{N}_{2}^{\mbf{A}}\circ\mc{N}_{1}^{\mbf{A}}\otimes \mrm{id}^{\mbf{B}}\right)\left(\rho^{\mbf{AB}}_{0,\src}\right),\left(\mc{N}^{\mbf{A}}_{2}\otimes\mc{S}^{\mbf{B}}_{1}\right)\left(\rho^{\mbf{AB}}_{1,\tgt}\right)\right)\nonumber\\
    &\quad\quad\quad\quad\quad\quad\quad\quad+\td\left(\left(\mc{N}^{\mbf{A}}_{2}\otimes\mc{S}^{\mbf{B}}_{1}\right)\left(\rho^{\mbf{AB}}_{1,\tgt}\right), \left(\mrm{id}^{\mbf{A}}\otimes\mc{S}^{\mbf{B}}_{1}\circ\mc{S}^{\mbf{B}}_{2}\right)\left(\rho^{\mbf{AB}}_{2,\tgt}\right)\right),
\end{align}
where we have defined a hybrid world channel $\mc{N}^{\mbf{A}}_{2}\otimes\mc{S}^{\mbf{B}}_{1}\colon\D(\mc{H}^{\mbf{AB}}_{1,\tgt})\to\D(\mc{H}^{\mbf{A}}_{2,\tgt}\otimes\mc{H}^{\mbf{B}}_{0,\src})$, and applied the triangle inequality. By data processing, we can further bound this inequality as
\begin{subequations}
    \begin{multline}
        \td\left(\left(\mc{N}_{2}^{\mbf{A}}\circ\mc{N}_{1}^{\mbf{A}}\otimes \mrm{id}^{\mbf{B}}\right)\left(\rho^{\mbf{AB}}_{0,\src}\right),\left(\mc{N}^{\mbf{A}}_{2}\otimes\mc{S}^{\mbf{B}}_{1}\right)\left(\rho^{\mbf{AB}}_{1,\tgt}\right)\right)\\
        \leq\td\left(\left(\mc{N}_{1}^{\mbf{A}}\otimes \mrm{id}^{\mbf{B}}\right)\left(\rho^{\mbf{AB}}_{0,\src}\right),\left(\id^{\mbf{A}}\otimes\mc{S}^{\mbf{B}}_{1}\right)\left(\rho^{\mbf{AB}}_{1,\tgt}\right)\right)\leq\delta_{\mbf{B},1},
    \end{multline}
    \begin{multline}
        \td\left(\left(\mc{N}^{\mbf{A}}_{2}\otimes\mc{S}^{\mbf{B}}_{1}\right)\left(\rho^{\mbf{AB}}_{1,\tgt}\right), \left(\mrm{id}^{\mbf{A}}\otimes\mc{S}^{\mbf{B}}_{1}\circ\mc{S}^{\mbf{B}}_{2}\right)\left(\rho^{\mbf{AB}}_{2,\tgt}\right)\right)\\
        \leq\td\left(\left(\mc{N}^{\mbf{A}}_{2}\otimes\id^{\mbf{B}}\right)\left(\rho^{\mbf{AB}}_{1,\tgt}\right), \left(\mrm{id}^{\mbf{A}}\otimes\mc{S}^{\mbf{B}}_{2}\right)\left(\rho^{\mbf{AB}}_{2,\tgt}\right)\right)\leq \delta_{\mbf{B},2},
    \end{multline}
\end{subequations}
demonstrating an additive upper bound on the composition. 
\qed

\section{\label{app:QSNIR-LUequi}Local Unitary Equivalence in QSNIRs}
When the local channel enacted by a party within a QSNIR is both deterministic and invertible (i.e. a reduction described by local unitary channels) we derive a simple result regarding the existence of a simulator proving privacy against the party. 
\begin{lemma}[Correctness Implies Privacy for Unitary QSNIRs]\label{lem:LU-QSNIR}
     Consider a pair of local unitary channels, where $\mbf{A}$ and $\mbf{B}$ apply unitaries $U$ and $V$, respectively, to their registers of $\rho^{\mbf{AB}}_{\src}$. Whenever $\epsilon$-correctness with respect to some $\rho^{\mbf{AB}}_{\tgt}$ is satisfied, $\epsilon=\delta_{\mbf{A}}=\delta_{\mbf{B}}$ privacy against $\mbf{A}$ and $\mbf{B}$ is also satisfied, by simulators  $U^{\dagger}$ and $V^{\dagger}$.
\end{lemma}
\begin{proof}
    Satisfying $\epsilon$ correctness requires that
    \begin{equation}
        \td\left(\left(U^{\mbf{A}}\otimes V^{\mbf{B}}\right)\rho^{\mbf{AB}}_{\src}\left((U^{\dagger})^{\mbf{A}}\otimes (V^{\dagger})^{\mbf{B}}\right),\rho^{\mbf{AB}}_{\tgt}\right)\leq\epsilon.
    \end{equation}
    Therefore, by data-processing the following must hold
    \begin{equation}
         \td\left(\left(U^{\mbf{A}}\otimes \mbb{I}^{\mbf{B}}\right)\rho^{\mbf{AB}}_{\src}\left((U^{\dagger})^{\mbf{A}}\otimes \mbb{I}^{\mbf{B}}\right),\left(\mbb{I}^{\mbf{A}}\otimes (V^{\dagger})^{\mbf{B}}\right)\rho^{\mbf{AB}}_{\tgt}\left(\mbb{I}^{\mbf{A}}\otimes V^{\mbf{B}}\right)\right)\leq\epsilon,
    \end{equation}
    as well, such that we identify a simulator, $\mc{S}^{\mbf{B}}(\cdot)=(V^{\dagger}\cdot V)^{\mbf{B}}$, satisfying $\epsilon$ privacy against $\mbf{B}$. An identical argument follows when $\mbf{A}$ is corrupt.
    \qed
\end{proof}
\noindent
While the above result is general to $\epsilon\geq0$, when $\epsilon=0$, $\rho^{\mbf{AB}}_{\src}$ and $\rho^{\mbf{AB}}_{\tgt}$ are said to be \textit{local unitary} (LU) \textit{equivalent}. In conjunction with Thm.~\ref{thm:sequential}, for any source state satisfying an $(\epsilon,\delta_{\mbf{A}},\delta_{\mbf{B}})$-QSNIR from some target, there must also exist a QSNIR, with the same security parameters, from the target to any LU equivalent source state.

\section{\label{app:QSNIR-SDP-dual}Deriving the Dual SDP for QSNIR privacy}
To reiterate, the optimal QSNIR privacy error for the extraction of a correlation from its embedding is given by the SDP:
\begin{equation}
    \begin{aligned}
        \textup{min} \quad & \frac{1}{2}\sum_{x}\tr[Q_{x}^{+}+Q_{x}^{-}], \\
        \textup{s.t.} \quad & Q_{x}^{+}-Q_{x}^{-}=p(x)\op{\psi_{x}}{\psi_{x}}-\sum_{y}p(x,y)J_{y}, & Q_{x}^{+},Q_{x}^{-}\succeq0, & \quad \forall x\in\mc{X}, \\
        & \tr[J_{y}]=1,
        & J_{y}\succeq0, & \quad \forall y\in\mc{Y}. \\
    \end{aligned}
\end{equation}
We start by defining dual variables $\{R_{x}\}_{x\in\mc{X}}$ and $\{\mu_{y}\}_{y\in\mc{Y}}$, such that $R_{x}=R_{x}^{\dagger}$ for all $x$, and $\mu_{y}\in\mbb{R}$ for all $y$. The Lagrangian is
\begin{multline}
    L  = \frac{1}{2}\sum_{x}\tr[Q_{x}^{+}+Q_{x}^{-}]  +\sum_{x}\tr\left[R_{x}\left(p(x)\op{\psi_{x}}{\psi_{x}}-\sum_{y}p(x,y)J_{y} -\left(Q_{x}^{+}-Q_{x}^{-}\right)\right)\right] \\ + \sum_{y}\mu_{y}(\tr[J_{y}]-1),
\end{multline}
which we rearrange as
\begin{multline}
    L = \sum_{x}p(x)\ip{\psi_{x}}{R_{x}|\psi_{x}} -\sum_{y} \mu_{y} +\sum_{x}\tr\left[\left(\frac{1}{2}\mbb{I}-R_{x}\right)Q_{x}^{+}+\left(\frac{1}{2}\mbb{I}+R_{x}\right)Q_{x}^{-}\right] \\
    +\sum_{y}\tr\left[\left(\mu_{y}\mbb{I}-\sum_{x}p(x,y)R_{x}\right)J_{y}\right].
\end{multline}
In order to force the dual variable to be positive semidefinite, we define a shifted operator $W_{x}\coloneqq R_{x}+\frac{1}{2}\mbb{I}$, and shifted scalar $\lambda_{y}\coloneqq \mu_{y}+\frac{1}{2}p(y)$. 

Consequently, we present the dual SDP as follows:
\begin{equation}
    \begin{aligned}
        \textup{max} \quad &  \sum_{x}p(x)\ip{\psi_{x}}{W_{x}|\psi_{x}}-\sum_{y}\lambda_{y},  \\
        \textup{s.t.} \quad & 0\preceq W_{x}\preceq \mbb{I}, &  & \quad \forall x\in\mc{X}, \\
        & \lambda_{y}\mbb{I}\succeq\sum_{x}p(x,y)W_{x}, & \lambda_{y}\in\mbb{R}, & \quad \forall y\in\mc{Y}.
    \end{aligned}
\end{equation}

\section{\label{app:lem-duality}Proof of Lem.~\ref{lem:QSNIRsdp-duality} (Strong Duality)}
We apply Slater's conditions to both sides of the QSNIR SDP pair. First for the primal, the optimum is $\delta^{(\textup{p})}\geq0$, as $Q_{x}^{+},Q_{x}^{-}\succeq0$ for all $x$, and is therefore finite. A strictly feasible point exists for the dual in setting $W_{x}=\mbb{I}/2$ for all $x$, and $\lambda_{y}=1$ for all $y$, for which the dual constraints are satisfied as
\begin{equation}
    \lambda_{y}\mbb{I}-\sum_{x}p(x,y)W_{x}=\left(1-\frac{p(y)}{2}\right)\mbb{I}\succeq \frac{1}{2}\mbb{I} \succeq 0.
\end{equation}

Likewise the dual optimum is finite and bounded from above by $\delta^{(\textup{d})}\leq1$, as 
 \begin{equation}
    \sum_{x}p(x)\ip{\psi_{x}}{W_{x}|\psi_{x}}\leq\sum_{x}p(x)=1,
\end{equation}
and
\begin{equation}
    \lambda_{y}\geq\lambda_{\max}\left(\sum_{x}p(x,y)W_{x}\right)\geq0,
\end{equation}
where $\lambda_{\max}(\cdot)$ is the maximum eigenvalue function.
We then determine a strictly feasible point in the primal problem. First we set $J_{y}=\mbb{I}/|\mc{Y}|$, the maximally mixed state, for all y, which is strictly positive. Next, we set
\begin{equation}
    Q_{x}^{-}=\mbb{I}, \quad Q_{x}^{+}=p(x)\left(\op{\psi_{x}}{\psi_{x}}-\frac{1}{|\mc{Y}|}\mbb{I}\right)+\mbb{I}, \quad \forall x.
\end{equation}
By construction this satisfies the corresponding primal constraint, so we only need to verify that $Q_{x}^{+}\succ 0$ for all $x$, where we confirm that
\begin{align}
     Q_{x}^{+} &\succeq  \left(1+\lambda_{\min}\left(p(x)\left(\op{\psi_{x}}{\psi_{x}}-\frac{1}{|\mc{Y}|}\mbb{I}\right)\right)\right)\mbb{I} \nonumber \\
     &\succ \left(1-\frac{p(x)}{|\mc{Y}|}\right)\mbb{I} \nonumber\\
     &\succ\left(1-\frac{1}{|\mc{Y}|}\right)\mbb{I},
\end{align}
where $\lambda_{\min}(\cdot)$ is the minimum eigenvalue function. Hence, $Q_{x}^{+}$ is strictly positive for all $x$, whenever $|\mc{Y}|\geq2$. 
\qed

\section{\label{app:comp-ent}Computing Entropies of Common/Dependent Variables for 2PC correlations}

\begin{table}[t]
    \centering
    \renewcommand{\arraystretch}{1.5}
    \begin{tabular}{|l||c||c|c|c|}
        \hline
        $(\X,\Y)$ & $I(\X\colon\Y)$ & $H(\X\wedge\Y)$ & $H(\Y\searrow\X|\X)$ & $H(\X\searrow\Y|\Y)$ \\
        \hline\hline
        $\msf{SK}$ & $1$ & $1$ &\multicolumn{2}{c|}{$0$} \\
        \hline
        $\msf{OK}$ & $1$ & $0$ & \multicolumn{2}{c|}{$1$} \\
        \hline\hline
        $\msf{BSC}_{q}$ & $1 - h(q)$ & $\begin{cases} 0, & q\in(0,1) \\ 1, & \textup{else}\end{cases}$ & \multicolumn{2}{c|}{$\begin{cases} 0, & q=1/2 \\ h(q), & \textup{else}\end{cases}$} \\
        \hline
        $\msf{BEC}_{q}$ & $1 - q$ & $\begin{cases} 1, & q=0 \\ 0, & \textup{else}\end{cases}$ & $h(q)$ & $\begin{cases} 0, & q=1 \\ q, & \textup{else}\end{cases}$ \\
        \hline
    \end{tabular}
    \setlength{\abovecaptionskip}{10pt}
    \caption{One-shot entropies for 2PC correlations. Cells for the conditional entropies are merged horizontally for correlations that are symmetric under a relabeling of $\X$ and $\Y$.}
    \label{tab:entropies}
\end{table}

In Tab.~\ref{tab:entropies}, we detail the zero-error entropies that satisfy monotones for privacy~\cite{Wolf-2005-NM}, for the correlations we study in Sec.~\ref{sec:2PC}. Specifically, we compute $I(\X\colon\Y)$ and $H(X\wedge\Y)$, whose difference is a monotone, and $H(\Y\searrow\X|\X)$ and $H(\X\searrow\Y|\Y)$, which are themselves monotones. When these quantities are $0$, the corresponding correlation admits a $0$-QSNIR to its embeddings (see Thm.~\ref{thm:perf-priv}). It will be important here to recall $h(q)\coloneqq -q\log q -(1-q)\log(1-q)$, as the binary entropy function. We also let $\delta_{x,y}$ denote the Kronecker delta function in this section, which is $1$ if $x=y$ and $0$ otherwise.

For symmetric key, $\msf{SK}$, we have that $\X=\Y=\X\wedge\Y$, such that perfect privacy is automatic. In contrast, for oblivious key, $\msf{OK}$, given by the pmf in Eq.~\ref{eq:ok-dist},
\begin{equation}
    I((\msf{W}_{0},\msf{W}_{1})\colon(\msf{V},\msf{S}))=H(\msf{W}_{0},\msf{W}_{1})+H(\msf{V},\msf{S})-H(\msf{OK}) = 2 + 2 - 3 = 1,
\end{equation}
and $H(\X\wedge\Y)=0$, as the connecting graph $G_{\msf{OK}}$ is fully connected (see Fig.~\ref{fig:OK-TRIP}). For the conditional zero-error entropies, we note that $(\msf{W}_{0},\msf{W}_{1})\searrow(\msf{V},\msf{S})$ is the distribution over the set of conditional random variables $\{(\msf{V},\msf{S})|(\msf{W}_{0},\msf{W}_{1})=(w_{0},w_{1})\}_{w_{0},w_{1}}$, which is uniform on the supported $(v,s)$ for each element. The function that maps $(w_0,w_1)$ to the set of supported $(v,s)$ is injective here, so $H((\msf{W}_{0},\msf{W}_{1})|(\msf{W}_{0},\msf{W}_{1})\searrow(\msf{V},\msf{S}))=0$. Therefore we can use the identity (Cor. 4 in~\cite{Wolf-2004-Zi}) that
\begin{equation}
    H(\X|\Y)=H(\X|\X\searrow\Y)+H(\X\searrow\Y|\Y),
\end{equation}
such that
\begin{equation}
    H((\msf{W}_{0},\msf{W}_{1})\searrow(\msf{V},\msf{S})|(\msf{V},\msf{S}))=H((\msf{W}_{0},\msf{W}_{1})|(\msf{V},\msf{S})) =1.
\end{equation}
The reverse entropy follows by symmetry.

For the family of correlations,  $\msf{BSC}_{q}$, we first verify that
\begin{equation}
    I(\X\colon\Y)=H(\X)+H(\Y)-H(\X,\Y)= 1+ 1- (1+h(q))=1-h(q).
\end{equation}
For $H(\X\wedge\Y)$, the connecting graph for $\msf{BSC}_{q}$ is fully connected for $q\in(0,1)$, such that $H(\X\wedge\Y)=0$ in this range. At end points the connecting graph is the same as for $\msf{SK}$, and $H(\X\wedge\Y)=1$ when $q\in\{0,1\}$. The conditional distribution, $\Y|\X=x$, is given by the pmf $p(y|x)\in\{1-q, q\}$ for either $x$, which leads exactly to $H(\Y|\X)=h(q)(=H(\X|\Y))$. However, the mapping from $y$ to $\X|\Y=y$ is injective when $q\neq1/2$, as $\msf{BSC}_{1/2}$ is a uniform product distribution. Therefore, 
\begin{equation}
    H(\X\searrow\Y|\Y)=H(\X|\Y)-H(\X|\X\searrow\Y)= h(q)-\delta_{q,1/2}.
\end{equation}
Again, the reverse entropy follows by symmetry.

For the family of correlations,  $\msf{BEC}_{q}$,
\begin{align}
     I(\X\colon\Y)&=H(\X)+H(\Y)-H(\X,\Y)\nonumber\\
     &= 1+ \left(-q\log q-(1-q)\log\left(\frac{1-q}{2}\right) \right) - (1+h(q)) \nonumber\\
     &=1-q,
\end{align}
where we made use of the fact that $q\log(q/2)-q\log(q)=-q$. As $\msf{BEC}_{0}=\msf{SK}$, $H(\X\wedge\Y)=1$ when $q=0$. For all $q\in(0,1]$ the connecting graph has only a single connected component (with additional disjoint vertices), such that $H(\X\wedge\Y)=0$ in this range. For the zero-error conditional entropies, we now have that $H(\X\searrow\Y|\Y)\neq H(\Y\searrow \X|\X)$. When $\mbf{A}$ is corrupt, we compute
\begin{equation}\label{eq:h-cond-BEC-A}
    H(\Y\searrow \X|\X) =H(\Y|\X)-H(\Y|\Y\searrow\X) =1+h(q)- 1 -0=h(q).
\end{equation}
Note that $\X|\Y=y$, when $y\in\{0,1\}$, is a point distribution, while $\X|\Y=\perp$ is always a uniform distribution. Therefore, the mapping from $y$ to $\X|\Y=y$ is injective for all $q$, leaving $H(\Y|\Y\searrow\X)=0$. Eq.~\ref{eq:h-cond-BEC-A} has some intuition behind it. When $q=1/2$ and $\mbf{A}$ receives $x$, then $\mbf{B}$ must have $y\in\{x,{\perp}\}$ with uniform probability. Instead, when $\mbf{B}$ is corrupt,
\begin{align}
    H(\X\searrow \Y|\Y) &=H(\X|\Y)-H(\X|\X\searrow\Y) \nonumber\\
    &=1+h(q)- \left(-q\log q-(1-q)\log\left(\frac{1-q}{2}\right) \right) -\delta_{q,1} \nonumber\\
    &=q  -\delta_{q,1}.
\end{align}
The key observation here is that $H(\X|\X\searrow\Y)=\delta_{q,1}$. Intuitively, when $q<1$ and we are given $x\samp\X$, we always have $y\in\{x,{\perp}\}$ with weights, $\{1-q,q\}$, such that there is no asymptotic uncertainty in $\X$ given $\X\searrow\Y$ and $H(\X|\X\searrow\Y)=0$. Conversely, when $q=1$, $y=\perp$ necessarily, which is independent of $x$, and therefore $H(\X|\X\searrow\Y)=1$, as $\X$ is a bit of uniform randomness.

\section{\label{app:thm-priv-OK}Proof of Thm.~\ref{thm:priv-OK} (Optimal Privacy Errors for $\msf{OK}$)}
We work with the representation of $\msf{OK}$ as  $\msf{Trip}\coloneqq ((\X_{0},\X_{1}),(\Y_{0},\Y_{1}))$. This correlation is fixed by the pmf
\begin{equation}
    p_{\msf{Trip}}(x_{0},x_{1},y_{0},y_{1})\coloneqq
    \begin{cases}
        1/8, & x_{0}\oplus y_{0}=x_{1}\wedge y_{1}, \\
        0, & \textup{else},
    \end{cases}
\end{equation}
which is symmetric under interchanging $\X$ and $\Y$. This symmetry is made even more evident when viewing the corresponding connecting graph for each correlation, depicted in Fig.~\ref{fig:OK-TRIP}. 
As this correlation is equivalent to $\msf{OK}$ via a $0$-SNIR,\footnote{Formally, the $0$-SNIR from $\msf{Trip}$ (also known as a Popescu-Rohrlich (PR)~\cite{Popescu-1994-Qn,Barrett-2005-Nc} correlation without inputs) to $\msf{OK}$ appears in \cite{Wolf-2005-Ot}, where it is a simple relabeling: $\mbf{A}$ sets $x_{0}\coloneqq w_{0}$ and $x_{1}\coloneqq w_{0}\oplus w_{1}$, and $\mbf{B}$ sets $y_{0}\coloneqq v$ and $y_{1}\coloneqq s$. As this reduction is perfectly private and involves only addition it is invertible as well.
} we again have that $I((\X_{0},\X_{1})\colon(\Y_{0},\Y_{1}))-H((\X_{0},\X_{1})\wedge(\Y_{0},\Y_{1}))=1$. 

To initiate the proof, it is useful to note a few details about the initial embedding
\begin{equation}
    \ket{\mrm{\Psi}_{\msf{Trip}}}^{\mbf{AB}}\coloneqq\sum_{x_0,x_1,y_0,y_1}\sqrt{p(x_{0},x_{1},y_{0},y_{1})}\ket{x_{0},x_{1}}^{\mbf{A}}\otimes\ket{y_{0},y_{1}}^{\mbf{B}}.
\end{equation}
The ensemble of reduced states itself can be represented explicitly in $\mbb{C}^{2}\otimes \mbb{C}^{2}$, or equivalently in $\mbb{C}^{4}$, as
\begin{subequations}\label{eq:ok-ens}
    \begin{align}
        \ket{\psi_{00}} \coloneqq \frac{1}{\sqrt{2}}\left(\ket{00}+\ket{01}\right)\left(\eqqcolon\ket{0{+}}\right) &\quad\Longleftrightarrow\quad \ket{\psi_{1}} \coloneqq \frac{1}{\sqrt{2}}\left(\ket{0}+\ket{1}\right), \\
        \ket{\psi_{01}} \coloneqq \frac{1}{\sqrt{2}}\left(\ket{00}+\ket{11}\right)\left(\eqqcolon \ket{\mrm{\Phi}^{+}}\right) &\quad\Longleftrightarrow\quad \ket{\psi_{0}} \coloneqq \frac{1}{\sqrt{2}}\left(\ket{0}+\ket{3}\right), \\
        \ket{\psi_{10}} \coloneqq \frac{1}{\sqrt{2}}\left(\ket{10}+\ket{11}\right)\left(\eqqcolon\ket{1{+}}\right)  &\quad\Longleftrightarrow\quad \ket{\psi_{3}} \coloneqq \frac{1}{\sqrt{2}}\left(\ket{2}+\ket{3}\right), \\
        \ket{\psi_{11}} \coloneqq \frac{1}{\sqrt{2}}\left(\ket{01}+\ket{10}\right)\left(\eqqcolon\ket{\mrm{\Psi}^{+}}\right) &\quad\Longleftrightarrow\quad \ket{\psi_{2}} \coloneqq \frac{1}{\sqrt{2}}\left(\ket{1}+\ket{2}\right),
    \end{align}
\end{subequations}
noting their familiarity as either products of qubit basis states or as Bell pairs in $\mbb{C}^{2}\otimes\mbb{C}^{2}$. Note that in the $\mbb{C}^{4}$ representation we have permuted the labels for each state, such that $\ket{\psi_{x}}=(\ket{x}+\ket{(x-1)\bmod4})/\sqrt{2}$.

The Gram matrix for this ensemble is 
\begin{equation}
    G\coloneqq \frac{1}{2}
    \begin{pmatrix}
    2 & 1 & 0 & 1 \\
    1 & 2 & 1 & 0 \\
    0 & 1 & 2 & 1 \\
    1 & 0 & 1 & 2
    \end{pmatrix},
\end{equation}
which has eigenvalues $\lambda_{G}=(2,1,0,1)$. Let $(c_{00},c_{01},c_{10},c_{11})\coloneqq(1,1/2,0,1/2)$, or $(c_{0},c_{1},c_{2},c_{3})\coloneqq(1,1/2,0,1/2)$.  Clearly $G$ is circulant, such that its rows are generated by a cyclic shifting of the elements in $(1,1/2,0,1/2)$ one to the right.
The elements of the Gram matrix are therefore
\begin{equation}
    G^{((x_{0},x_{1}),(x_{0}',x_{1}'))}=c_{x_{0}\oplus x_{0}',x_{1}\oplus x_{1}'}, \quad\Longleftrightarrow\quad G^{(x,x')} =c_{(x'-x)\bmod 4},
\end{equation}
in either representation. Hence the symmetry in $G$ is generated by an Abelian group over $\mbb{Z}_{2}\times\mbb{Z}_{2}$, or alternatively the cyclic group over $\mbb{Z}_{4}$, which demonstrates why the Gram matrix is circulant. These symmetries suffice to show that this ensemble is geometrically uniform~\cite{Eldar-2001-qd}. 
Let $X_{d}\coloneqq\sum_{x=0}^{d-1}\op{(x+1)\bmod d}{x}$ be the cyclic shift operator of dimension, $d$. The unitaries which generate each group are
\begin{equation}
    \left\{U_{x_{0},x_{1}}\coloneqq \textup{C}X_{1\to0}^{x_{1}}(X^{x_{0}}\otimes\mbb{I}) \right\}_{x_{0},x_{1}} \quad\Longleftrightarrow\quad\left\{U^{x}\colon U\coloneqq X_{4} \right\}_{x}
\end{equation}
where we recognize $X\coloneqq X_{2}$ as the Pauli-$X$ gate, and 
$\textup{C}X_{1\to0}$ as the controlled-$X$ gate (i.e. CNOT). It is straightforward to verify that
\begin{equation}
    \ket{\psi_{x_{0}x_{1}}}=U_{x_{0},x_{1}}\ket{\psi_{00}} \quad\Longleftrightarrow\quad \ket{\psi_{x}}=U^{x}\ket{\psi_{0}},
\end{equation}
demonstrating the geometric uniformity of the ensemble in either representation.

To determine $\delta_{\textup{MED}}^{\msf{Trip}}$, the above implies that the PGM is optimal for MED privacy, through Lem.~\ref{lem:PGM}. In the ideal setting,
\begin{equation}
    P_{\textup{guess}}((\X_{0},\X_{1})|(\Y_{0},\Y_{1}))_{\textup{ideal}}=\frac{1}{2}.
\end{equation}
In order to compute the optimal minimum-error guessing probability we utilized the eigenvalues of the Gram matrix, such that
\begin{equation}
    P_{\textup{guess}}((\X_{0},\X_{1})|\mbf{B})_{\textup{real}}=\frac{1}{4^{2}}\left(\sum_{x_{0},x_{1}}\sqrt{\lambda_{G}^{(x_{0},x_{1})}}\right)^{2}=\frac{1}{8}\left(3+2\sqrt{2}\right).
\end{equation}
Putting this together, we have
\begin{equation}
    \delta_{\textup{MED}}^{\msf{Trip}}= P_{\textup{guess}}((\X_{0},\X_{1})|\mbf{B})_{\textup{real}}-P_{\textup{guess}}((\X_{0},\X_{1})|(\Y_{0},\Y_{1}))_{\textup{ideal}}=\frac{1}{8}\left(2\sqrt{2}-1\right)
\end{equation}
as the optimal MED privacy error for $\msf{OK}$, when extracting this correlation from its embedding.

For computing $\delta_{\textup{QSNIR}}^{\msf{Trip}}$, we work in the $\mbb{C}^{4}$ representation. Note that with these labels the correlation is given by a parent pmf
\begin{equation}
    p(x,y)=
    \begin{cases}
        1/8,  & (x-y)\bmod4\in\{0,1\} \\
        0, & \textup{else}
    \end{cases}
\end{equation}
accounting for the shift we applied in Eq.~\ref{eq:ok-ens}. For any dual feasible $\{W_{x}\}_{x}$, we have that
\begin{equation}
    \delta_{\textup{QSNIR}}^{\msf{Trip}}\geq \frac{1}{4}\sum_{x=0}^{3}\ip{\psi_{x}}{W_{x}|\psi_{x}}-\frac{1}{8}\sum_{y=0}^{3}\lambda_{\max}\left(W_{y}+W_{(y+1)\bmod4}\right).
\end{equation}
For a nonnegative $t\geq0$, let
\begin{equation}
    W_{x}\coloneqq \op{w_{x}}{w_{x}},\quad\ket{w_{x}}\coloneqq\frac{1}{\sqrt{1+t^{2}}}\left(\ket{\psi_{x}}-t\ket{\psi_{(x-2)\bmod4}}\right),
\end{equation}
define a family of dual feasible operators, as each operator is the projector onto a valid pure state. By construction, the set of pure states, $\{\ket{w_{x}}\}_{x}$, has the same cyclic symmetry, where $U^{x}\ket{w_{0}}=\ket{w_{x}}$.
Note that $\ip{w_{x}}{\psi_{x}}=1/\sqrt{1+t^{2}}$, such that
\begin{equation}
    \frac{1}{4}\sum_{x=0}^{3}\ip{\psi_{x}}{W_{x}|\psi_{x}}=\frac{1}{1+t^{2}}.
\end{equation}
Furthermore, each $W_{y}+W_{(y+1)\bmod4}$ has identical eigenvalues of
\begin{equation}
    \lambda_{W}=\left(\frac{3}{2}-\frac{t}{1+t^{2}},\frac{1}{2}+\frac{t}{1+t^{2}},0,0\right),
\end{equation}
listed in nonincreasing order for all $t\geq0$. 
Therefore,
\begin{equation}
    \frac{1}{8}\sum_{y=0}^{3}\lambda_{\max}\left(W_{y}+W_{(y+1)\bmod4}\right)=\frac{3}{4}-\frac{t}{2(1+t^{2})},
\end{equation}
and the total bound becomes
\begin{equation}
     \delta_{\textup{QSNIR}}^{\msf{Trip}}\geq\frac{1}{1+t^{2}}-\left(\frac{3}{4}-\frac{t}{2(1+t^{2})}\right)=\frac{2+t}{2(1+t^{2})}-\frac{3}{4}\eqqcolon f(t).
\end{equation}
We can maximize this bound by taking
\begin{equation}
    f'(t)=\frac{-t^{2}-4t+1}{2(1+t^{2})^{2}},
\end{equation}
and noting that the polynomial, $t^{2}+4t-1=0$, has roots, $-2\pm\sqrt{5}$. As $t\geq 0$, we take the positive root, verify that $f''(-2+\sqrt{5})=-1-\frac{9}{4\sqrt{5}}<0$, and the bound becomes
\begin{equation}
    \delta_{\textup{QSNIR}}^{\msf{Trip}}\geq f(-2+\sqrt{5})=\frac{1}{4}\left(\sqrt{5}-1\right).
\end{equation}

We now construct a set of states prepared by the simulator, $\{J_{y}\}_{y=0}^{3}$, which achieves this bound, proving that it is in fact the optimum. Using the symmetry in the ensemble let us rewrite the primal objective as
\begin{equation}
    \delta_{\mc{S}}=\frac{1}{4}\sum_{x=0}^{3}\td\left(\psi_{x},\frac{1}{2}\left(J_{x}+J_{(x-1)\bmod4}\right)\right)=\td\left(\psi_{0},\frac{1}{2}\left(J_{0}+J_{3}\right)\right),
\end{equation}
where we have recognized that each term in the sum minimizes identically. Let $J_{y}\coloneqq \op{s_{y}}{s_{y}}$. From Lem.~\ref{lem:QSNIRsdp-duality}, complementary slackness suggests that these states should satisfy
\begin{equation}
   (W_{y}+W_{(y+1)\bmod 4}) \ket{s_{y}}=\left(\frac{3}{2}-\frac{t}{1+t^{2}}\right)\ket{s_{y}}, \quad \forall y.
\end{equation}
We define,
\begin{align}
     \ket{s_{y}}&\coloneqq \frac{1}{\sqrt{2\lambda_{\textup{max}}}}\left(\ket{w_{y}}+\ket{w_{(y+1)\bmod4}}\right) \nonumber\\
     &=\frac{1}{\sqrt{2\lambda_{\textup{max}}(1+t^{2})}}(\ket{\psi_{y}}-t\ket{\psi_{(y-1)\bmod4}}-t\ket{\psi_{(y-2)\bmod4}}+\ket{\psi_{(y-3)\bmod 4}}),
\end{align}
as a set of normalized pure states, where
\begin{equation}
    \lambda_{\max}\coloneqq \frac{3}{2}-\frac{t}{1+t^{2}}.
\end{equation}
We next rewrite
\begin{align}
     J_{0}+J_{3} &= \op{s_{0}}{s_{0}}+\op{s_{3}}{s_{3}} \nonumber\\
     &=\frac{1}{2}\left((\ket{s_{0}}+\ket{s_{3}})(\bra{s_{0}}+\bra{s_{3}})+(\ket{s_{0}}-\ket{s_{3}})(\bra{s_{0}}-\bra{s_{3}})\right) \nonumber\\
     &=(1+z)\op{s_{0,+}}{s_{0,+}}+(1-z)\op{s_{0,-}}{s_{0,-}}
\end{align}
having further defined the orthonormal pair
\begin{equation}
    \ket{s_{0,+}}\coloneqq\frac{1}{\sqrt{2(1+z)}}\left(\ket{s_{0}}+\ket{s_{3}}\right), \quad \ket{s_{0,-}}\coloneqq\frac{1}{\sqrt{2(1-z)}}\left(\ket{s_{0}}-\ket{s_{3}}\right),
\end{equation}
with $z\coloneqq \ip{s_{0}}{s_{3}}$.\footnote{A corresponding pair of orthogonal states similarly exists for each $J_{x}+J_{(x-1)\bmod 4}$ in the objective.} 
This is useful as, explicitly,
\begin{equation}
    \ket{s_{0,-}}=\frac{1}{2}\left((\ket{0}-\ket{3})+(\ket{1}-\ket{2})\right), 
\end{equation}
such that $\ip{\psi_{0}}{s_{0,-}}=0$. This also follows from the symmetry that $\ip{\psi_{0}}{s_{0}}=\ip{\psi_{0}}{s_{3}}\eqqcolon a$.
Therefore,
\begin{align}
    \delta_{\mc{S}} &= \frac{1}{2}\left\|\op{\psi_{0}}{\psi_{0}}-\frac{1}{2}\left((1+z)\op{s_{0,+}}{s_{0,+}}+(1-z)\op{s_{0,-}}{s_{0,-}}\right)\right\|_{1}  \nonumber\\
    &=\frac{1-z}{4}+\frac{1}{2}\left\|\op{\psi_{0}}{\psi_{0}}-\frac{1}{2}\left((1+z)\op{s_{0,+}}{s_{0,+}}\right)\right\|_{1} \nonumber\\
    &=\frac{1-z}{4}+\frac{1}{2}\sqrt{\left(1+\frac{1+z}{2}\right)^{2}-2(1+z)|\ip{\psi_{0}}{s_{0,+}}|^{2}} \nonumber\\
    &=\frac{1-z}{4}+\frac{1}{2}\sqrt{\left(\frac{3+z}{2}\right)^{2}-4a^2}.
\end{align}
Finally, plugging in $t=-2+\sqrt{5}$,
\begin{equation}
    z=\frac{2(t-1)^{2}}{3t^{2}-2^{t}+3}=\frac{1}{11}(7-\sqrt{5}),\quad a=\frac{3-t}{2\sqrt{2\lambda_{\max}(1+t^{2})}}=\frac{1}{2}\sqrt{\frac{5}{11}(4+\sqrt{5})},
\end{equation}
and after some further computation
\begin{equation}
    \frac{1-z}{4}+\frac{1}{2}\sqrt{\left(\frac{3+z}{2}\right)^{2}-4a^2}=\frac{1}{44}(4+\sqrt{5})+\frac{5}{44}(2\sqrt{5}-3)=\frac{1}{4}(\sqrt{5}-1).
\end{equation}
The set $\{\op{s_{y}}{s_{y}}\}_{y=0}^{3}$ is therefore primal feasible and achieves the lower bound that we determined from the dual problem. In total this proves that $\delta_{\textup{QSNIR}}^{\msf{Trip}}=(\sqrt{5}-1)/4$.
\qed

\section{\label{app:thm-BSC-priv}Proof of Thm.~\ref{thm:BSC-priv} (Optimal Privacy Errors for $\msf{BSC}$)}
Before proving anything, note that $\msf{BSC}_{q}$ is symmetric under interchange of the labels for $\X$ and $\Y$. Therefore, any nonzero optimal privacy errors are equivalent for $\mbf{A}$ and $\mbf{B}$. Here, when relevant, we reference the case of a corrupt $\mbf{B}$.

Without loss of generality let $q\leq1/2$. 
When one of the parties is honest, the other's reduced state of the embedding is described by an ensemble consisting of
\begin{equation}
    \ket{\psi_{0}}=\sqrt{1-q}\ket{0}+\sqrt{q}\ket{1}, \quad\quad \ket{\psi_{1}}=\sqrt{q}\ket{0}+\sqrt{1-q}\ket{1},
\end{equation}
each occurring with marginal probability $p(0)=p(1)=1/2$. On the Bloch sphere we can represent these states by Bloch vectors $r_{0}=(s,0,t)$ and $r_{1}=(s,0,-t)$, where $s=2\sqrt{q(1-q)}$ and $t=1-2q$. Note that purity requires that $s^{2}+t^{2}=1$. Note also that this ensemble is geometrically uniform, and generated by $X\coloneqq \op{0}{1}+\op{1}{0}$, the shift operator in $\mbb{C}^{2}$ (i.e. the Pauli-$X$ gate). That is, $X\ket{\psi_{0}}=\ket{\psi_{1}}$, explicitly.

First, we demonstrate that $\delta_{\textup{MED}}^{\msf{BSC}}(q)=0$. Given ideal samples of $\msf{BSC}_{q}$, we have that
\begin{equation}
    P_{\textup{guess}}(\X|\Y)_{\textup{ideal}}=\max\{q,1-q\}.
\end{equation}
Intuitively this is because the classical minimum-error guessing strategy for the corrupt party is to guess $\textup{argmax}_{x}p(x,y)$. When extracting this correlation from its entangled embedding, we can compute the optimal MED guessing probability from Lem.~\ref{lem:helevo} as 
\begin{align}
    P_{\textup{guess}}(\X|\mbf{B})_{\textup{real}}
    &=\frac{1}{2}\left(1+\td\left(\psi_{0},\psi_{1}\right)\right) \nonumber\\
    &=\frac{1}{2}\left(1+\frac{1}{2}\|r_{0}-r_{1}\|_{2}\right) \nonumber\\
    &= \frac{1}{2}\left(1+|t|\right) \nonumber\\
    &=\max\{q,1-q\},
\end{align}
demonstrating its equality with the classical analog, such that $\delta_{\textup{MED}}^{\msf{BSC}}(q)=0$.

We now compute $\delta_{\textup{QSNIR}}^{\msf{BSC}}(q)$, which is a nontrivial function. We define a simulator that outputs one of the states, $\{J_{0},J_{1}\}$, with Bloch vectors $j_{0}=(u,0,v)$ and $j_{1}=(u,0,-v)$, respectively, where $u^{2}+v^{2}\leq1$. We fix the solution to the $XZ$-plane of the Bloch sphere, without loss of generality, and the symmetry of these vectors enforces that $XJ_{0}X=J_{1}$, just as it is the case with the states in the ensemble.
On average, the adversary's state in the ideal world is one of
\begin{equation}
    \sigma_{0}=(1-q)J_{0} +qJ_{1}, \quad\quad \sigma_{1}=qJ_{0}+(1-q)J_{1},
\end{equation}
with Bloch vectors
\begin{equation}
     \hat{r}_{0}=(1-q)j_{0}+qj_{1}=(u,0,tv), \quad\quad  \hat{r}_{1}=qj_{0}+(1-q)j_{1}=(u,0,-tv),
\end{equation}
respectively.
Therefore, we can rewrite the primal objective function as
\begin{align}
    \delta_{\mc{S}} &= \frac{1}{2}\bigg(\td\left(\psi_{0},\sigma_{0}\right)+\td\left(\psi_{1},\sigma_{1}\right)\bigg) \nonumber\\
    &=\frac{1}{4}\left(\left\| r_{0}-\hat{r}_{0}\right\|_{2}+\left\| r_{1}-\hat{r}_{1}\right\|_{2}\right) \nonumber\\
    &= \frac{1}{2}\left\| r_{0}-\hat{r}_{0}\right\|_{2} \nonumber\\
    &=\frac{1}{2}\sqrt{(s-u)^{2}+t^{2}(1-v)^{2}}.
\end{align}
Note that independent of $u$ and $v$, $\delta_{\mc{S}}=0$, when $q\in\{0,1/2,1\}$. In total, we have parametrized the problem into the following convex optimization:
\begin{equation}
    \delta_{\textup{QSNIR}}^{\msf{BSC}}(q)=\min_{\mc{S}}\delta_{\mc{S}}=\frac{1}{2}\min_{u^{2}+v^{2}\leq1}\bigg[\|(s-u,0,t(1-v))\|_{2}\bigg].
\end{equation}
To constrain the problem further we recognize that, for a fixed value $u$, the objective is minimized as $v^{2}\to1-u^{2}$. Therefore, the optimum always lies on the boundary of the unit disk, i.e. when the simulator prepares pure states. We can further set $u=\sin\theta$ and $v=\cos\theta$, with $\theta\in[0,\pi/2]$, recovering
\begin{equation}
     \delta_{\textup{QSNIR}}^{\msf{BSC}}(q)=\frac{1}{2}\min_{0\leq\theta\leq\pi/2}\bigg[\sqrt{(s-\sin\theta)^{2}+t^{2}(1-\cos\theta)^{2}}\bigg].
\end{equation}
While no longer convex, the minimization now describes the problem of finding the closest point from $(s,t)$ that lies on the ellipse formed by $\sin^{2}\theta +t^{-2}\cos^{2}\theta=1$. This problem has no closed-form solution that holds for every $q\in[0,1]$, but can readily be optimized numerically, in order to produce the corresponding curve in Fig.~\ref{fig:BSC-BEC}.
\qed

\section{\label{app:them-BEC-priv}Proof of Thm.~\ref{thm:BEC-priv} (Optimal Privacy Errors for $\msf{BEC}$)}
$\msf{BEC}_{q}$ is asymmetric such that we expect different privacy errors when either $\mbf{A}$ or $\mbf{B}$ is corrupt. The entropies in Tab.~\ref{tab:entropies} reveal that perfect privacy holds for both parties, only at the endpoints, when $q\in\{0,1\}$. Below we consider each corrupt case as an individual proof.

\subsection{\label{app:them-BEC-priv-A}Corrupt $\mbf{A}$}
When $\mbf{A}$ is corrupt, their reduced state is one of
\begin{equation}
    \ket{\psi_{0}}=\ket{0}, 
        \quad\quad \ket{\psi_{1}}=\ket{1}, \quad\quad
        \ket{\psi_{\perp}}=\ket{{+}}\coloneqq \frac{1}{\sqrt{2}}(\ket{0}+\ket{1}),
\end{equation}
with priors $p(0)=p(1)=(1-q)/2$ and $p(\perp)=q/2$.
It will also be useful to specify $\ket{-}\coloneq(\ket{0}-\ket{1})/\sqrt{2}$. Note that the average state of the ensemble is
\begin{equation}
    \rho=\frac{1-q}{2}\left(\op{0}{0}+\op{1}{1}\right)+q\op{+}{+},
\end{equation}
which is invariant under $X\coloneqq \op{0}{1}+\op{1}{0}$, such that $X\rho X=\rho$.

As usual, we start by computing $\delta_{\textup{MED},\mbf{A}}^{\msf{BEC}}(q)$.
In the ideal world $\mbf{A}$ receives $x\samp\X$. The optimal ideal minimum-error discrimination strategy yields a success probability
\begin{equation}
    P_{\textup{guess}}(\Y|\X)_{\textup{ideal}}=\max\{1-q, q\},
\end{equation}
which follows from the intuition that if $q\leq 1/2$ $\mbf{A}$ should guess that $y=x$, and $q> 1/2$, they should guess $y=\perp$. In the real world, we define a POVM $\{M_{0}, M_{1}, M_{\perp}\}$, such that
\begin{equation}\label{eq:psucc-bec-a}
    P_{\textup{succ}}(\Y|\mbf{A})_{\textup{real}}=\frac{1-q}{2}\left(\ip{0}{M_{0}|0}+\ip{1}{M_{1}|1}\right)+q\ip{{+}}{M_{\perp}|{+}},
\end{equation}
which in general may not be optimal. 

Consider a candidate POVM
\begin{equation}
    M_{0}\coloneqq\op{0}{0}, \quad\quad M_{1}\coloneqq\op{1}{1},\quad\quad M_{\perp}\coloneqq0,
\end{equation}
where the outcome $\perp$ never occurs. Of course this is just a computational basis measurement, and plugging this into the objective yields
\begin{equation}\label{eq:psucc-bec-a-3}
     P_{\textup{succ}}(\Y|\mbf{A})_{\textup{real}}=1-q.
\end{equation}
We check the Kennedy-Yuen-Lax equations from Lem.~\ref{lem:kyl}, which are defined by constraints
\begin{equation}
    \frac{1-q}{2}(\mbb{I}-\op{0}{0})\succeq0, \quad\quad \frac{1-q}{2}(\mbb{I}-\op{1}{1})\succeq0, \quad\quad \frac{1-q}{2}\mbb{I}-q\op{+}{+}\succeq0.
\end{equation}
The first two constraints are trivially always satisfied, as both operators have eigenvalues $((1-q)/2,0)$, however, the operator in the third constraint has eigenvalues $((1-q)/2,(1-3q)/2)$. Therefore, this PVM is in fact optimal when $q\leq1/3$, and $P_{\textup{guess}}(\Y|\mbf{A})_{\textup{real}}$ is given by Eq.~\ref{eq:psucc-bec-a-3}. Note that in this range $\delta_{\textup{MED},\mbf{A}}^{\msf{BEC}}(q)=1-q-\max\{q,1-q\}=0$.

When $q>1/3$, we start with a family of POVMs
\begin{subequations}\label{eq:povm-bec-a-med}
    \begin{align}
        M_{0}(z)&\coloneqq\frac{1}{2}\left(z\ket{+}+\ket{-}\right)\left(z\bra{+}+\bra{-}\right), \\
        M_{1}(z)&\coloneqq\frac{1}{2}\left(z\ket{+}-\ket{-}\right)\left(z\bra{+}-\bra{-}\right), \\
        M_{\perp}(z)&\coloneqq(1-z^{2})\op{+}{+},
    \end{align}
\end{subequations}
where $z\in[0,1]$. Note this POVM is complete by construction, as
\begin{equation}
     M_{0}(z)+M_{1}(z)=z^{2}\op{+}{+}+\op{-}{-}.
\end{equation}
Eq.~\ref{eq:psucc-bec-a} then becomes a second-order polynomial in $z$,
\begin{align}
    P_{\textup{succ}}(\Y|\mbf{A})_{\textup{real}}&=\frac{1-q}{4}(z+1)^{2}+q(1-z^{2}) \nonumber\\
    &=-\frac{5q-1}{4}z^{2}+\frac{1-q}{2}z+\frac{1+3q}{4}, \nonumber\\
    &=-\frac{5q-1}{4}\left(z-\frac{1-q}{5q-1}\right)^{2}+\frac{4q^{2}}{5q-1}
\end{align}
where in the last line we completed the square. In the regime where $1/3<q\leq1$, $5q-1>0$ such that maximum success probability is $4q^{2}/(5q-1)$, exactly when $z=(1-q)/(5q-1)\in[0,1]$, which is feasible.

To prove optimality of this POVM in this regime, we construct an upper bound $P_{\textup{guess}}(\Y|\mbf{A})_{\textup{real}}\leq4q^{2}/(5q-1)$ from the dual SDP for MED, in Eq.~\ref{eq:sdp-med-gen}. Consider the Hermitian operator
\begin{equation}
    \mrm{\Lambda}\coloneqq q\op{+}{+}+r\op{-}{-},
\end{equation}
where we choose $r\geq0$. The dual constraints are
\begin{equation}
    \mrm{\Lambda}-\frac{1-q}{2}\op{0}{0}\succeq0, \quad\mrm{\Lambda}-\frac{1-q}{2}\op{1}{1}\succeq0,\quad \mrm{\Lambda}-q\op{+}{+}\succeq0.
\end{equation}
The third constraint trivially always holds, as $\mrm{\Lambda}-q\op{+}{+}=r\op{-}{-}\succeq0$. For the first constraint, the diagonal components of the operator in the computational basis are $\{2q+r-1,q+r\}$, and its determinant is
\begin{equation}
    \det\left[\mrm{\Lambda}-\frac{1-q}{2}\op{0}{0}\right]=\frac{1}{4}\left(r(5q-1)-q(1-q)\right),
\end{equation}
Sylvester's criterion states that this operator is positive semi-definite if and only if all three of these quantities are nonnegative. Therefore, $\mrm{\Lambda}$ is dual feasible if
\begin{equation}
    q+r\geq 0, \quad r \geq1-2q, \quad r\geq \frac{q(1-q)}{5q-1},
\end{equation}
which identically is the set of conditions for the second constraint. The first inequality always holds, and $q(1-q)/(5q-1) \geq1-2q$, when $1/3<q\leq1$. Therefore, feasibility reduces to $r\geq q(1-q)/(5q-1)$. Choosing the minimum value we compute the dual optimum as
\begin{equation}
    \tr[\mrm{\Lambda}]=q+\frac{q(1-q)}{5q-1}=\frac{4q^{2}}{5q-1},
\end{equation}
demonstrating that the POVM in Eq.~\ref{eq:povm-bec-a-med} is optimal (given the right $z$) when $1/3<q\leq 1$, and that for the full range of $q$ we have
\begin{equation}
    \delta^{\msf{BEC}}_{\textup{MED},\mbf{A}}(q) \coloneqq \begin{cases}
        0, &  0 \leq q \leq 1/3,\\
    4q^2/(5q-1)-\max\{q,1-q\}, & 1/3 < q \leq 1.
    \end{cases}     
\end{equation}
as the optimal MED privacy error when $\mbf{A}$ is corrupt.
    
Proving the corresponding QSNIR privacy error for this side of $\msf{BEC}_{q}$ follows similarly to the proof for $\msf{BSC}_{q}$.
Let $\{J_{0},J_{1}\}$ be the states prepared by the simulator given $x\samp\X$, such that the average state of the adversary in the ideal world is
\begin{equation}
    \sigma_{0}=J_{0}, \quad \sigma_{1}=J_{1}, \quad \sigma_{\perp}=\frac{1}{2}\left(J_{0}+J_{1}\right).
\end{equation}
To respect the symmetry of the real world ensemble, we require that $XJ_{0}X=J_{1}$. Furthermore, the Bloch vectors of the states in the real ensemble are
\begin{equation}
    r_{0}=(0,0,1),\quad r_{1}=(0,0,-1),\quad r_{\perp}=(1,0,0),
\end{equation}
such that we choose $\{j_{0},j_{1}\}$, the Bloch vectors for $\{J_{0},J_{1}\}$, to lie in the $XZ$-plane, without loss of generality. Hence, we can parameterize $\{j_{0},j_{1}\}$ as
\begin{equation}
    j_{0}=(u,0,v),\quad j_{1}=(u,0,-v),
\end{equation}
where $u,v$ are real, such that $u^{2}+v^{2}\leq 1$.
The primal objective computing QSNIR privacy is now
\begin{align}
    \delta_{\mc{S}} &= \frac{1-q}{2}\left(\td\left(\psi_{0},J_{0}\right)+\td\left(\psi_{1},J_{1}\right)\right)+  q\td\left(\psi_{\perp},\frac{1}{2}(J_{0}+J_{1})\right) \nonumber\\
    &= \frac{1}{2}\left(\frac{1-q}{2}\left(\left\|r_{0}-j_{0}\right\|_{2}+\left\|r_{1}-j_{1}\right\|_{2}\right)+q\left\|r_{\perp}-\frac{1}{2}(j_{0}+j_{1})\right\|_{2}\right) \nonumber\\
    &=\frac{1}{2}\left((1-q)\sqrt{u^{2}+(1-v)^{2}}+ q\sqrt{(1-u)^{2}}\right),
\end{align} 
where in the last line we used that $\|r_{0}-j_{0}\|_{2}=\|r_{1}-j_{1}\|_{2}$. To minimize the first term in this objective, we choose $v^{2}=1-u^{2}$, i.e. a solution on the boundary of the unit disk, and set $u=\sin\theta$ and $v=\cos\theta$ for $\theta\in[0,\pi/2]$. The objective can be recast as the minimization problem
\begin{equation}
    \delta^{\msf{BEC}}_{\textup{QSNIR},\mbf{A}}(q) = \min_{0\leq\theta\leq\pi/2}\left[(1-q)\sin\frac{\theta}{2}+\frac{q}{2}(1-\sin\theta)\right].
\end{equation}

Unlike the case with $\msf{BSC}_{q}$, the scalar minimization above has a closed-form solution, albeit a cluttered one. Let
\begin{equation}
    f'_{q}(\theta)\coloneqq \frac{1}{2}\left((1-q)\cos\frac{\theta}{2}-q\cos\theta\right)
\end{equation}
denote the first derivative of the minimization argument above.
We first use that $\cos(\theta/2)\geq\cos\theta\geq0$ on the whole interval $\theta\in[0,\pi/2]$, such that
\begin{equation}
      f'_{q}(\theta)\geq \frac{1-2q}{2}\cos\theta\ge0,
\end{equation}
when $q\in[0,1/2]$. The minimum in this region is then $f_{q}(0)=q/2$. When instead $q\in(1/2,1]$, solving explicitly for when $f'_{q}(\theta_{0})=0$ yields
\begin{equation}
    \theta_{0}\coloneqq 4\tan^{-1}\sqrt{3q-\sqrt{1+q(9q-2)}}.
\end{equation}
This critical point corresponds to a minimum, as $f''_{q}(\theta_{0})\geq0$, when $q\in(1/2,1]$. In total
\begin{equation}
    \delta^{\msf{BEC}}_{\textup{QSNIR},\mbf{A}}(q) =
    \begin{cases}
        q/2, & 0\leq q \leq1/2,\\
        (1-q)\sin\frac{\theta_{0}}{2}+\frac{q}{2}(1-\sin\theta_{0}), & 1/2< q\leq 1.
    \end{cases}
\end{equation}
\qed

\subsection{\label{app:them-BEC-priv-B}Corrupt $\mbf{B}$}
When $\mbf{B}$ is corrupt, their reduced state is a uniform ensemble over 
\begin{equation}
    \ket{\psi_{0}}=\sqrt{1-q}\ket{0}+\sqrt{q}\ket{{\perp}}, \quad\quad \ket{\psi_{1}}=\sqrt{1-q}\ket{1}+\sqrt{q}\ket{{\perp}},
\end{equation}
which are each elements of $\mbb{C}^{3}$. Note that the overlap between these states is $q$. It will be useful to symmetrize these states as
\begin{equation}
    \ket{\psi_{0}} = \sqrt{\frac{1+q}{2}}\ket{s}+\sqrt{\frac{1-q}{2}}\ket{t}, \quad\quad \ket{\psi_{1}} = \sqrt{\frac{1+q}{2}}\ket{s}-\sqrt{\frac{1-q}{2}}\ket{t}.
\end{equation}
Here we have defined a pair of orthogonal states,
\begin{equation}
      \ket{s} \coloneqq \sqrt{\frac{1-q}{2(q+1)}}\left(\ket{0}+\ket{1}\right)+\sqrt{\frac{2q}{1+q}}\ket{{\perp}}, \quad\quad \ket{t} \coloneqq \frac{1}{\sqrt{2}}(\ket{0}-\ket{1}),
\end{equation}
where $\ket{s}$ represents the part of the ensemble that is symmetric under permuting $\X$ and the part of $\Y$ containing $\{0,1\}$, and $\ket{t}$ is the antisymmetric part.

As $\mbf{B}$ must discriminate between only two states, we can apply Lem.~\ref{lem:helevo} to compute the MED privacy error. 
Ideally, given a sample $y\samp\Y$,
\begin{equation}
    P_{\textup{guess}}(\X|\Y)_{\textup{ideal}}=1 - \frac{q}{2},
\end{equation}
where the ideal minimum-error guessing probability decreases linearly to the uniform distribution. We then apply Lem.~\ref{lem:helevo}, and compute
\begin{align}
    P_{\textup{guess}}(\X|\mbf{B})_{\textup{real}}
    &=\frac{1}{2}\left(1+\td\left(\psi_{0},\psi_{1}\right)\right) \nonumber\\
    &=\frac{1}{2}\left(1+\sqrt{1-|\ip{\psi_{0}}{\psi_{1}}|^{2}}\right) \nonumber\\
    &=\frac{1}{2}+\frac{1}{2}\sqrt{1-q^{2}},
\end{align}
such that in total
\begin{equation}
    \delta_{\textup{MED},\mbf{B}}^{\msf{BEC}}(q)=\frac{1}{2}\left(\sqrt{1-q^{2}}-(1-q)\right).
\end{equation}

While Lem.~\ref{lem:helevo} is sufficient for computing $\delta_{\textup{MED},\mbf{B}}^{\msf{BEC}}(q)$, it's instructive for the rest of the proof to determine the Helstrom measurement, a binary PVM, $\{M,\mbb{I}-M\}$, that achieves the optimum. Given the symmetric representation of $\mbf{B}$'s ensemble, we have that
\begin{equation}
    \op{\psi_{0}}{\psi_{0}} - \op{\psi_{1}}{\psi_{1}} = \sqrt{1-q^{2}}\left(\op{s}{t}+\op{t}{s}\right)\eqqcolon \mrm{\Delta}. 
\end{equation}
The eigenvectors of $\mrm{\Delta}$ are $\ket{\phi_{+}}\coloneqq(\ket{s}+\ket{t})/\sqrt{2}$ and  $\ket{\phi_{-}}\coloneqq(\ket{s}-\ket{t})/\sqrt{2}$, with eigenvalues $\pm\sqrt{1-q^{2}}$ respectively, which are the basis for the projectors onto the positive and negative eigenspaces of $\mrm{\Delta}$. Note that
\begin{equation}
    \ip{\phi_{\pm}}{\psi_{0}} = \frac{1}{2}\left(\sqrt{1+q}\pm\sqrt{1-q}\right)=\ip{\phi_{\mp}}{\psi_{1}}.
\end{equation}
Hence, it is further evident that the PVM $\{\op{\phi_{+}}{\phi_{+}},\op{\phi_{-}}{\phi_{-}}\}$ can discriminate $\{\ket{\psi_{0}},\ket{\psi_{1}}\}$ with nontrivial probability.

We now verify that this choice of PVM achieves the optimal guessing probability. In general,
\begin{equation}
    P_{\textup{guess}}(\X|\mbf{B})_{\textup{real}}=\frac{1}{2}\left(\tr\left[M \op{\psi_{0}}{\psi_{0}}\right]+\tr\left[(\mbb{I}-M) \op{\psi_{1}}{\psi_{1}}\right]\right) = \frac{1}{2} + \frac{1}{2}\tr\left[M\mrm{\Delta}\right].
\end{equation}
Plugging in $M=\op{\phi_{+}}{\phi_{+}}$, 
\begin{equation}
     P_{\textup{guess}}(\X|\mbf{B})_{\textup{real}}=\frac{1}{2}+\frac{1}{2}\ip{\phi_{+}}{\mrm{\Delta}|\phi_{+}}=\frac{1}{2}+\frac{1}{2}\sqrt{1-q^{2}}.
\end{equation}

In order to compute the optimal QSNIR error, we define $\{J_{0},J_{1},J_{\perp}\}$ as the set of states prepared by the simulator, and
\begin{equation}
    \sigma_{0} \coloneqq (1-q)J_{0}+qJ_{\perp}, \quad\quad \sigma_{1} \coloneqq (1-q)J_{1}+qJ_{\perp}
\end{equation}
are the states that we must minimize the trace distance over, when taken against the ensemble in the real world. That is,
\begin{equation}
    \delta_{\mc{S}}=\frac{1}{2}\left(\td\left(\psi_{0},\sigma_{0}\right)+\td\left(\psi_{1},\sigma_{1}\right)\right).
\end{equation}
Immediately, from Lem.~\ref{lem:min-err-LB}, $\delta_{\mc{S}}\geq\delta_{\textup{MED},\mbf{B}}^{\msf{BEC}}$, for all $\mc{S}$. Now let
\begin{equation}
    J_{0}\coloneqq \op{\phi_{+}}{\phi_{+}},\quad\quad J_{1}\coloneqq\op{\phi_{-}}{\phi_{-}}, \quad\quad J_{\perp}\coloneqq\op{s}{s},
\end{equation}
such that
\begin{equation}
    \sigma_{0} = (1-q)\op{\phi_{+}}{\phi_{+}}+q\op{s}{s}, \quad\quad \sigma_{1} = (1-q)\op{\phi_{-}}{\phi_{-}}+q\op{s}{s}.
\end{equation}
Explicitly,
\begin{align}
    \op{\psi_{0}}{\psi_{0}}-\sigma_{0}=\frac{1}{2}\left(\sqrt{1-q^{2}}-(1-q)\right)\left(\op{s}{t}+\op{t}{s}\right)\eqqcolon\mrm{\tilde{\Delta}}_{0},
\end{align}
and likewise $ \op{\psi_{1}}{\psi_{1}}-\sigma_{1}\eqqcolon\mrm{\tilde{\Delta}}_{1}=-\mrm{\tilde{\Delta}}_{0}$. 
We recognize that the coefficient above is exactly $\delta_{\textup{MED},\mbf{B}}^{\msf{BEC}}(q)$, which we momentarily refer to as just $\delta$. Both $\mrm{\tilde{\Delta}}_{0}$ and $\mrm{\tilde{\Delta}}_{1}$ have identical eigenvalues of $(\delta,-\delta,0)$. Hence,
\begin{equation}
    \delta_{\mc{S}}=\frac{1}{4}\left(\left(|\delta|+|-\delta|\right)+\left(|\delta|+|-\delta|\right)\right)=\delta,
\end{equation}
the lower bound is tight, and $\delta_{\textup{QSNIR},\mbf{B}}^{\msf{BEC}}(q)=\delta_{\textup{MED},\mbf{B}}^{\msf{BEC}}(q)\eqqcolon\delta_{\mbf{B}}^{\msf{BEC}}(q)$.
\qed

\section{\label{app:BEY}Embeddings With Encoded Phase Information}
Throughout this work, we operated entirely under the assumption that the embeddings, $\ket{\mrm{\Psi}_{\X\Y}}$, we used were restricted to those where $\varphi(x,y)=0$ for all $x$ and $y$. This follows from an intuition, motivated by the qubit case, that these phases should only increase the distance between the pure states in the adversary's reduced state ensemble, $\{\{p(x),\ket{\psi_{x}}\}\}_{x}$. We now work without this assumption.

We first show that perfect privacy now further constrains the phases encoded in the embedding.
\begin{corollary}[to Thm.~\ref{thm:perf-priv}]\label{cor:perf-priv-phase}
     A reduction from $(\X,\Y)$ to its embedding $\ket{\mrm{\Psi}_{\X\Y}}^{\mbf{AB}}$ (containing encoded phase information given by $\varphi(x,y)$), has perfect privacy (against either a corrupt $\mbf{A}$ or $\mbf{B}$), if and only if $H(\Y\searrow\X|\X)=0$, \textup{and} $\varphi(x,y')-\varphi(x',y')$ is constant in $y'$, for any $x,x'$ that are supported on the same $y$.
\end{corollary}
\begin{proof}
    As is the case in Thm.~\ref{thm:perf-priv}, perfect privacy implies that if there exists any $x,x'$ that are supported with any single $y$, then (now with phases included)
    \begin{equation}
        1=|\ip{\psi_{x}}{\psi_{x'}}|=\left|\sum_{y'}\sqrt{p(y'|x)p(y'|x')}e^{i(\varphi(x,y')-\varphi(x',y'))}\right|.
    \end{equation}
    This imposes the additional constraint on the relative phases.
\end{proof}
\noindent
Furthermore, the structure of the optimal simulator in Thm~\ref{thm:perf-priv} continues to hold, as it has no dependence on $\varphi(x,y)$.

When perfect security is impossible from Cor.~\ref{cor:perf-priv-phase}, the objective in the primal within Eq.~\ref{eq:sdp-qsnir} is no longer convex with encoded phase information as a variable in the problem. Indeed this is a problem that we do not fully resolve here. We do, however, provide two insights below. 

\subsection{\label{sec:BEY-qubit}The Qubit Case}
For the case when the adversary's ensemble elements are qubit pure states, $\ket{\psi_{0}},\ket{\psi_{1}}\in\mbb{C}^{2}$, we seek to motivate our intuition that encoded phase information only helps the adversary. Certainly, with respect to MED privacy between two states, $\{\ket{\psi_{0}},\ket{\psi_1}\}$, it is known that the success probability depends on only a single angle between the states~\cite{Barnett-2009-Qs}. We demonstrate a source for this intuition in the context of QSNIR privacy as well.
\begin{lemma}\label{lem:qubit-case}
    For any $(0,\delta_{\mbf{A}},\delta_{\mbf{B}})$-QSNIR from $(\X,\Y)$ to $\ket{\mrm{\Psi}_{\X\Y}}^{\mbf{AB}}$, where $|\mc{X}|=|\mc{Y}|=2$ (such that $\dim\mc{H}^{\mbf{A}}=\dim\mc{H}^{\mbf{B}}=2$),  $(\delta_{\mbf{A}},\delta_{\mbf{B}})$ are upper-bounded by a quantity that is minimized when $\varphi(x,y)=0$, for all $x$ and $y$.
\end{lemma}
\begin{proof}
    Let $p_{xy}\coloneqq p(x,y)$ for all $x,y\in\mbb{Z}_{2}$, and $p_{x}\coloneqq \sum_{y}p_{xy}$ be its marginals.
    Generically, the subnormalized reduced state of the adversary in the real world, conditioned on $x\samp\X$, is
    \begin{subequations}
        \begin{align}
            \ket{\tilde{\psi}_{0}}&\coloneqq\sqrt{p_{0}}\ket{\psi_{0}}= \sqrt{p_{00}}\ket{0}+\sqrt{p_{01}}\ket{1}, \\  \ket{\tilde{\psi}_{1}}&\coloneqq \sqrt{p_{1}}\ket{\psi_{1}} =  \sqrt{p_{10}}\ket{0}+e^{i\varphi}\sqrt{p_{11}}\ket{1},
        \end{align}
    \end{subequations}
    where $\varphi\coloneqq -\varphi(0,0)+\varphi(0,1)+\varphi(1,0)-\varphi(1,1)$, removing any global phase differences between the two states. 
    Suppose a simulator exists, which prepares $\{J_{0},J_{1}\}$ and satisfies $\delta$-privacy for this ensemble, as
    \begin{equation}
        \delta=\td\left(\tilde{\psi}_{0},p_{00}J_{0}+p_{01}J_{1}\right)+\td\left(\tilde{\psi}_{1},p_{10}J_{0}+p_{11}J_{1}\right).
    \end{equation}

    Let $S_{-\varphi}\coloneqq\op{0}{0}+e^{-i\varphi}\op{1}{1}$, which when applied to $\ket{\psi_{1}}$ sends $\varphi\to0$. It would be nice if we could simply demonstrate that having the simulator produce $\{J_{0}, S_{-\varphi}J_{1}(S_{-\varphi})^{\dagger}\}$ generates the same $\delta$ as with a version of the ensemble where $\varphi=0$. This only works for certain correlations, such as $\msf{SK}$, where the environment compares the exact state produced by the simulator to its corresponding real world ensemble element. In general the ideal world averages over the states produced by the simulator, such that we can't update the channel with a unitary on only \textit{some} of the inputs and achieve the same error.\footnote{In contrast, while we can always define a global unitary that takes the ensemble back to the $XZ$-plane, this modifies the complex amplitudes of the state in the computational basis, such that a secure reduction to this rotated ensemble must come from a different distribution than what was originally embedded.}

    Instead, we analyze the privacy error more directly via the Bloch sphere.
    The corresponding (subnormalized) Bloch vector for each state in the real world ensemble is
    \begin{subequations}
        \begin{align}
            r_{0} &\coloneqq \left(2\sqrt{p_{00}p_{01}},0,p_{00}-p_{01}\right), \\
            r_{1} &\coloneqq \left(2\sqrt{p_{10}p_{11}}\cos(\varphi),2\sqrt{p_{10}p_{11}}\sin(\varphi),p_{10}-p_{11}\right).
        \end{align}
    \end{subequations}
    Let $n_{x}\coloneqq r_{x}/p_{x}$ be the corresponding normalized Bloch vector.
    Note that the distance $\|n_{0}-n_{1}\|_{2}$ is minimized when $\varphi=0$ for any correlation.\footnote{This is proven simply by the law of cosines, where one shows that $\|n_{0}-n_{1}\|_{2}^{2}=a-b\cos\varphi$, with $a$ and $b$ as nonnegative constants determined by the distribution.}
    In the ideal world, the simulator prepares states with Bloch vectors $\{j_{0},j_{1}\}$, respectively. The average subnormalized state of the adversary in the ideal world is then given by either of the Bloch vectors
    \begin{equation}
        \hat{r}_{0}\coloneqq p_{00}j_{0} + p_{01}j_{1}, \quad \hat{r}_{1}\coloneqq p_{10}j_{0} + p_{11}j_{1},
    \end{equation}
    such that the QSNIR privacy error we aim to minimize is
    \begin{align}
        2\delta &=\left\|r_0-\hat{r}_{0}\right\|_{2}+\left\|r_1-\hat{r}_{1}\right\|_{2} \nonumber\\
        &=\left\|r_{0}-(p_{00}j_{0}+p_{01}j_{1})\right\|_{2} +\left\|r_{1}-(p_{10}j_{0}+p_{11}j_{1})\right\|_{2} \nonumber\\        
        &\leq p_{00}\left\|n_{0}-j_{0}\right\|_{2}+p_{01}\left\|n_{0}-j_{1}\right\|_{2} +p_{10}\left\|n_{1}-j_{0}\right\|_{2}+p_{11}\left\|n_{1}-j_{1}\right\|_{2} \nonumber\\
        &=\sum_{y}p_{0y}\left\|n_{0}-j_{y}\right\|_{2}+p_{1y}\left\|n_{1}-j_{y}\right\|_{2}.
    \end{align}
    In the third line, we have bounded $\delta$ by a pair of distances to be minimized, over $j_{0}$ and $j_{1}$, respectively. 
    Since $\{p_{0y},p_{1y}\}$ are nonnegative, a value of $j_{y}$ that minimizes the corresponding summand is the vector in $\{n_{0},n_{1}\}$ that carries the larger weight.\footnote{Note that this minimization is a weighted (two-point) Fermat--Weber problem. If the weights are equal, then every point on the line segment between $n_{0}$ and $n_{1}$ is an optimal solution.} Therefore,
    \begin{equation}
        \min_{j_{y}} p_{0y}\left\|n_{0}-j_{y}\right\|_{2}+p_{1y}\left\|n_{1}-j_{y}\right\|_{2}=\min\{p_{0y},p_{1y}\}\left\|n_{0}-n_{1}\right\|_{2},
    \end{equation}
    and the total bound is
    \begin{equation}
     2\delta \leq \left(\min\{p_{00},p_{10}\}+\min\{p_{01},p_{11}\}\right)\left\|n_{0}-n_{1}\right\|_{2},
    \end{equation}
    which is minimized when $\varphi=0$.
\end{proof}
\noindent
While this lemma isn't sufficient to prove that the error is minimized without phase information (as its not a lower bound), it does suggest that constructing the optimal simulator is related to finding the geometric median between the weighted Bloch vectors of the real world ensemble elements. Clearly, there exist some distributions where this bound is tight, and the privacy error is then minimized without phases.

Importantly, despite the bound, this lemma shows why even in the qubit case, proving that these phases can only decrease privacy is nontrivial. The ideal world averages over the output of the simulator, such that for a general distribution, we can't just update the simulator with a unitary that cancels out any relative phase. Still, there is nonetheless structure here which may be employed to further prove our conjecture. 

\subsection{\label{sec:BEY-num}Numerical Minimization over Phase Information}
We briefly report numerical results, minimizing over encoded phase information, for $\msf{OK}$, $\msf{BSC}_{1/4}$, $\msf{BEC}_{1/2}$ when $\mbf{A}$ is corrupt, and $\msf{BEC}_{1/\sqrt{2}}$ when $\mbf{B}$ is corrupt.\footnote{For the correlation families we are interested in, we chose the points where the corresponding QSNIR privacy error is roughly maximized, in hopes we would see a greater contrast.}
These results suggest the optima in Tab.~\ref{tab:privacy-errors}
are minimal over all phases. We constructed a numerical optimizer in cvxpy, utilizing an open-source SCS solver, for the SDP in Eq.~\ref{eq:sdp-qsnir}, and sampled $10000$ random phase functions $\varphi(x,y)$ for each correlation. Each numerical minimum over all phase functions fell within the solver's tolerance of the analytically computed optimum for $\varphi(x,y)=0$. Code can be found attached as supplemental material.

\end{document}